\documentclass[reprint,aip,jmp,a4paper,onecolumn]{revtex4-1}
\pdfoutput=1

\usepackage{amsmath}
\usepackage{amsthm}
\usepackage{amssymb}
\usepackage{amscd}

\usepackage{todonotes}
\usepackage{xcolor}
\newcommand{\red}{\color{black}}

\usepackage{hyperref}

\usepackage{url}

\usepackage{fancyhdr}

\usepackage[T1]{fontenc} 
\usepackage{textcomp} 
\usepackage{lmodern} 

\usepackage[utf8]{inputenc}
        
\usepackage{enumitem} 

\usepackage{graphicx}
\usepackage{xcolor}
\usepackage{float}

\graphicspath{{figures/}}
\DeclareGraphicsExtensions{.jpg,.png,.pdf,.ai}
\newtheorem{lemma}{Lemma}
\newtheorem{theorem}[lemma]{Theorem}
\newtheorem{proposition}[lemma]{Proposition}
\newtheorem{corollary}[lemma]{Corollary}

\theoremstyle{definition}
\newtheorem{example}[lemma]{Example}
\newtheorem{definition}[lemma]{Definition}

\newtheorem{remark}[lemma]{Remark}

\newcommand{\setR}{\mathbb{R}}

\newcommand{\setT}{\mathbb{T}}
\newcommand{\setS}{\mathbb{S}}

\newcommand{\dd}{\mathrm{d}}

\renewcommand{\Im}{\mathrm{Im}}

\DeclareMathOperator{\Arg}{Arg}

\newcommand{\EM}{F}

\newcommand{\rel}{\text{rel}}
\newcommand{\glob}{\text{glob}}
\DeclareMathOperator{\var}{Var}
\newcommand{\closure}[1]{\mathrm{cl}(#1)}

\begin{document}

\title{Rotation Forms and Local Hamiltonian Monodromy}

\author{K.~Efstathiou}
\email{K.Efstathiou@rug.nl}
\affiliation{Johann Bernoulli Institute for Mathematics and Computer Science,
  University of Groningen, P.O. Box 407, 9700 AK, Groningen, The Netherlands}

\author{A.~Giacobbe}
\email{giacobbe@dmi.unict.it}
\address{Universit\`a di Catania, Dipartimento di Matematica e Informatica, Viale A. Doria 6, 95125 Catania, Italy}

\author{P.~Marde\v{s}i\'c}
\email{mardesic@u-bourgogne.fr}
\affiliation{Universit\'e de Bourgogne, Institut de Math\'ematiques de
  Bourgogne - UMR 5584, UFR Sciences et Techniques, 9, Avenue Alain Savary, BP
  47870, 21078 Dijon CEDEX, France}

\author{D.~Sugny}
\email{dominique.sugny@u-bourgogne.fr}
\affiliation{Laboratoire Interdisciplinaire Carnot de Bourgogne (ICB), UMR 6303 CNRS-Universit\'e de Bourgogne-Franche Comt\'e, 9 Av. A. Savary, BP 47 870, F-21078 Dijon Cedex, France}

\begin{abstract}
  The monodromy of torus bundles associated to completely integrable systems can be computed using geometric techniques (constructing homology cycles) or analytic arguments (computing discontinuities of abelian integrals). In this article we give a general approach to the computation of monodromy that resembles the analytical one, reducing the problem to the computation of residues of polar 1-forms. We apply our technique to three celebrated examples of systems with monodromy (the champagne bottle, the spherical pendulum, the hydrogen atom) and to the case of non degenerate focus-focus singularities, re-obtaining the classical results. An advantage of this approach is that the residue-like formula can be shown to be local in a neighborhood of a singularity, hence allowing the definition of monodromy also in the case of non-compact fibers. This idea has been introduced in the literature under the name of scattering monodromy. We prove the coincidence of the two definitions with the monodromy of an appropriately chosen compactification.
\end{abstract}

\maketitle

\section{Introduction}
\label{sec/intro}

A Liouville-Arnold integrable system is a map $\EM$ (called the \emph{map of integrals of motion} or \emph{integral map}) from a $2n$-dimensional symplectic manifold $M$ to $\mathbb R^n$ such that the components $\EM_j$, $j=1,\dots,n$, of $\EM$ Poisson commute. Let $\mathcal R$ denote a connected component of the set of regular values of $\EM$ and $\mathcal M$ denote a connected component of the preimage $\EM^{-1}(\mathcal R)$. Assuming that the level sets of $\EM$ are compact, the Liouville-Arnold theorem \cite{Arnold1989} states that $\EM: \mathcal M \to \mathcal R$ is a $\mathbb T^n$-bundle over $\mathcal R$. If $\mathcal R$ is not simply connected, then the $\mathbb T^n$-bundle $\EM|_\Gamma$ over a simple closed path $\Gamma$ in $\mathcal R$ may have non-trivial monodromy. Equivalently, there are no smooth action variables throughout $\mathcal R$ \cite{Nekhoroshev1972, Duistermaat1980}. In $n = 2$ degree of freedom systems with a circle action, monodromy can be identified with an integer number. If the number $n$ of degrees of freedom is larger than $2$, then $\mathcal R$ could possibly have non-trivial second cohomology. In that case the Liouville-Arnold integrable system could have global action variables but have non-trivial Chern class or, equivalently, no corresponding global angle variables which together with the action variables give a symplectic chart \cite{Nekhoroshev1972, Duistermaat1980}.

Non-trivial monodromy has been shown to exist in several integrable Hamiltonian systems such as the spherical pendulum~\cite{Duistermaat1980, Cushman2015}, the champagne bottle~\cite{Bates1991}, and the hydrogen atom in crossed electric and magnetic fields~\cite{Cushman2000}. In the mid-90's it was realized that a common property of these systems was the existence of isolated, focus-focus, critical values in the image of $\EM$. The presence of such focus-focus critical values causes a non-trivial fundamental group, $\pi_1(\mathcal R)$, and it turns out that the corresponding $\mathbb T^2$ bundle $\EM|_\Gamma$ over a path $\Gamma$ in $\mathcal R$ encircling the critical value has non-trivial monodromy \cite{Lerman1994, Matveev1996, Zung1997}. This result, now referred to as the \emph{geometric monodromy theorem}, has been further generalized to the non-Hamiltonian context \cite{Zung2002, Cushman2001}.

In this paper we focus on $2$ degree of freedom systems where $F = (H,J)$ are smooth. The function $H$ is the Hamiltonian of a Hamiltonian vector field $X_H$, while $J$ is the momentum of a Hamiltonian $\mathbb S^1$ action whose infinitesimal action is $X_J$. Establishing the non-triviality of monodromy along a closed path $\Gamma$ in such systems is often done through the study of the variation of the \emph{rotation number} along $\Gamma$. We give the definition of the rotation number in Section~\ref{sec/rotation}, see Definition~\ref{rotnum},  where we discuss in detail how the non-trivial variation of the rotation number along $\Gamma$ is equivalent to the non-trivial monodromy of the $\mathbb T^2$ bundle over $\Gamma$. We only note here that the definition of the rotation number is based on a geometric construction but its computation is typically done through the evaluation of an (abelian) integral and the investigation of its dependence upon the values $(h,j)$ of the integrals of motion. Moreover, the variation of the rotation number has been used to describe fractional monodromy \cite{Efstathiou2007, Sugny2008} and to define scattering monodromy \cite{Bates2007a}.

In the present work we relate the proofs of the non-triviality of monodromy based on the variation of the rotation number to a more geometric approach. In particular, we formalize an analytical computation of the rotation number through the notion of \emph{rotation 1-form} (Definition \ref{defrotform}), a closed 1-form whose integral over suitably defined orbit-segment of $X_H$ gives the rotation number up to a term which we prove to be unimportant for the variation. Moreover, we show that the variation is independent of the choice of the rotation 1-form provided that the latter satisfies a transversality condition (Definition \ref{deftrans}).

It turns out that a rotation 1-form cannot be defined in the whole phase-space, but it must necessarily be singular on a subset, whose points we call \emph{poles}. Such subset is essential for the non-triviality of monodromy. In all examples known to the authors, the set of poles is a 2-dimensional submanifold intersecting $\EM^{-1}(\Gamma)$ at a finite number of $X_J$-orbits, cf.\ Section~\ref{sec/examples}. The main result in this article is the following theorem relating the analytic computation of the variation of the rotation number to the geometry of the set of poles of the rotation 1-form.

\begin{theorem}\label{thm/main} Consider a two-degree of freedom integrable Hamiltonian system $F$, such that the fibers of $F$ are compact and connected. Consider a closed path $\Gamma$ in the set of regular values of $F$ and assume that there is a neighborhood $U$ of $F^{-1}(\Gamma)$ where $F$ is invariant under a Hamiltonian $\mathbb{S}^1$ action generated by a momentum $J$. Let $\vartheta$ be a rotation 1-form for the vector field $X_J$, transversal to $F$, and let $\Pi$ be its polar locus, which we assume two-dimensional. Further, assume that $\Gamma$ transversally intersects $F(\Pi)$ at a finite number of values $v_i$. Then the poles of the rotation 1-form in $F^{-1}(\Gamma)$ are a disjoint union of a finite number of $X_J$-orbits $\mathbb S^1 p_{ij} \in F^{-1}(v_i)$, which we call \emph{polar orbits}, and the monodromy number $k$ along $\Gamma$, see Eq.~\eqref{varTheta}, is given by
    \begin{equation}\label{k}
      k = \frac{1}{2\pi} \sum_{ij} \int_{\delta_{ij}} \vartheta,
    \end{equation}
    where $\delta_{ij}$ is a loop in $F^{-1}(\Gamma)$ surrounding $\mathbb S^1 p_{ij}$ with appropriate orientation, see Figure~\ref{fibration}.
\end{theorem}

\begin{remark}
  Theorem \ref{thm/main} applies to any torus bundle, provided that a Hamiltonian circle action, leaving $F$ invariant, is defined in a neighborhood of the torus bundle. The Theorem reduces the problem of computing the variation of the rotation number to that of integrating the rotation $1$-form $\vartheta$ along closed paths encircling the poles of $\vartheta$. A method for constructing the rotation $1$-form $\vartheta$ is given in Lemma \ref{lem-rot}. The integrals $\frac{1}{2\pi} \int_{\delta_{ij}} \vartheta$ are real analogues of residues for the rotation $1$-form around its set of poles. This is strongly reminiscent of the complex analytic approach of Ref.~\citenum{Sugny2008}, where the variation of the rotation number is expressed as the integral around the pole(s) of a meromorphic $1$-form.
\end{remark}

\begin{figure*}
\includegraphics[width=0.75\textwidth]{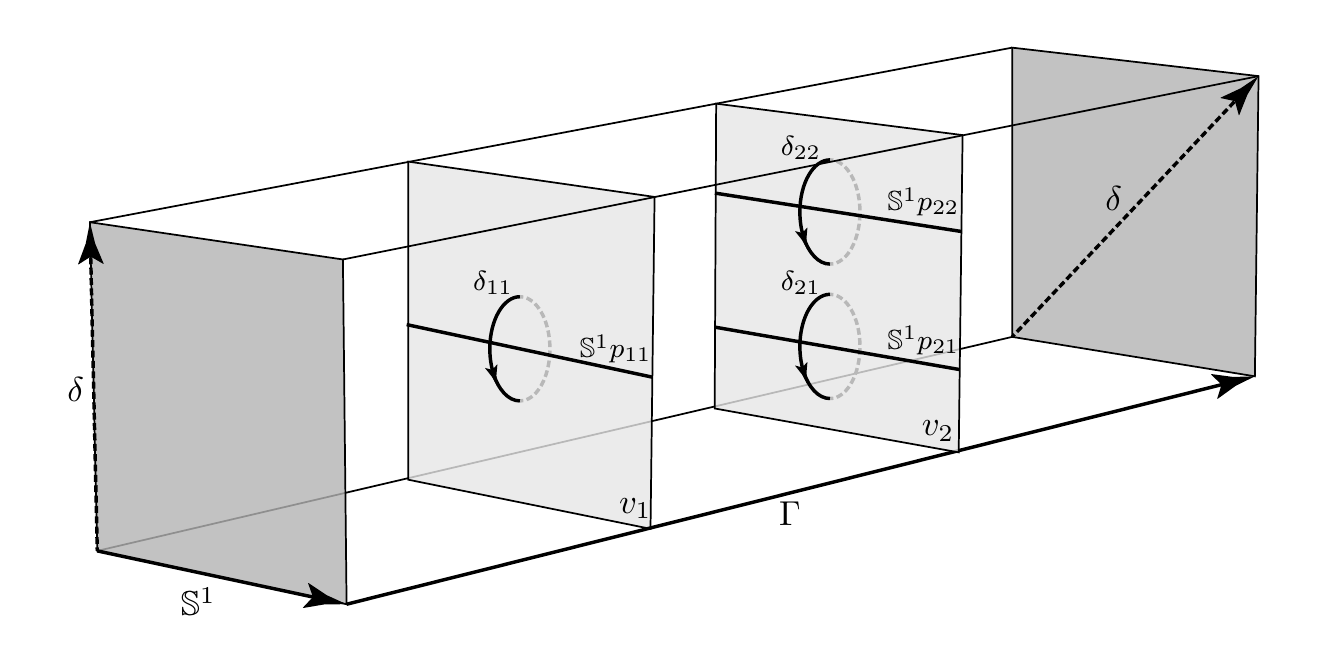}
\caption{The fibration $\EM$ above $\Gamma$. The circle $\Gamma$ and the fibers $\EM^{-1}(v) \simeq \mathbb T^2$ are unfolded for easier presentation. Polar orbits $\mathbb S^1 p_{ij}$ and integration loops $\delta_{ij}$ are shown for fibers $\EM^{-1}(v_i)$, see Theorem~\ref{thm/main}. The fiber $\EM^{-1}(\Gamma(0)) = \EM^{-1}(\Gamma(1))$ is represented by the dark gray surfaces. One should pay attention to that the identification of $\EM^{-1}(\Gamma(0))$ and $\EM^{-1}(\Gamma(1))$ is not be the one implied by this unfolded representation of $\EM^{-1}(\Gamma)$ when the monodromy is non-trivial. To highlight this we draw a representative of a homology cycle $\delta$ on $F^{-1}(\Gamma(0))$ and a possible representative of the \emph{same} cycle on $F^{-1}(\Gamma(1))$.}
\label{fibration}
\end{figure*}

The local form of an integrable Hamiltonian system in a neighborhood of a focus-focus critical point allows to apply Theorem~\ref{thm/main} and obtain the following well known fact.

\begin{corollary}
  Let $p$ be a focus-focus critical point of $F$ and $\Gamma$ a simple closed path in the set of regular values of $F$, such that $p$ is the only critical point in $\EM^{-1}(D)$, where $D$ is the set bounded by $\Gamma$. Then the monodromy number along $\Gamma$ is $k=-1$.
\end{corollary}

The main contribution of this paper does not lie in the computation of the monodromy, but in the systematic approach to monodromy through the variation of the rotation number and the expression of the latter as the integral of a rotation $1$-form. More specifically, the monodromy number is given by the sum of the integral of the rotation $1$-form along the cycles $\delta_{ij}$ described in Theorem~\ref{thm/main}. Applying this approach to the case of focus-focus points yields as a consequence that the Hamiltonian monodromy relies only on the local structure of the foliation in a neighborhood of such points (cf.~similar \emph{local} approaches in Ref.~\citenum{Schmidt2010} and Ref.~\citenum{Vu-Ngoc1998}). We preferred to present here the method in the easiest case of single focus-focus points and plan to apply it to more complicated cases, where additional difficulties appear, in a forthcoming work. In particular, we plan to deal with cases of non-isolated singularities such as the $(m:n)$-resonance case. In such general cases, more complicated contribution given by Picard-Lefschetz formula can appear.

Moreover, understanding how monodromy is locally determined in the case of single focus-focus points permits a generalization of the notion of monodromy to completely integrable Hamiltonian systems \emph{having not necessarily compact fibers}, avoiding the, frequently artificial, compactification of the fibers by adding suitable higher order terms to the Hamiltonian. We compare our local approach to monodromy based on the rotation $1$-form to the notion of \emph{scattering monodromy} introduced in Ref.~\citenum{Bates2007a}. We show that the two concepts are similar, and we highlight the role played by the identification of incoming and outgoing asymptotic directions in scattering monodromy.

The plan of the paper is as follows. In Section~\ref{sec/rotation} we give the definition of the rotation number and describe how the non-triviality of its variation is related to the non-triviality of monodromy. Then, we introduce rotation 1-forms and we show how they can be used to determine the variation of the rotation number. In Section~\ref{sec/examples} we give several examples of rotation 1-forms in specific examples of integrable Hamiltonian systems. In Section~\ref{sec/local} we study focus-focus singularities and show that the variation of the rotation number can be computed through the variation of an appropriate integral of a locally defined rotation $1$-form.
In Section~\ref{sec/scattering} we define monodromy for non-compact fibrations and relate our results to scattering monodromy. We draw conclusions and give perspectives in Section~\ref{sec/discussion}.

\section{Monodromy and Rotation Number}
\label{sec/rotation}

As stated in the Introduction, in this work we restrict our attention to 2 degree of freedom integrable systems ($n=2$) under the very typical hypothesis that one of the integrals of motion is a function $J$ which is the momentum of a circle action $\mathbb S^1 \times \mathcal M \to \mathcal M$, $(t,p) \to s \cdot p$, with $t \in [0,2\pi]$ (where $\mathcal M$ is possibly an open subset of the phase space). The other integral of motion is an $\mathbb S^1$-invariant function $H$ typically called Hamiltonian or energy of the system. For this reason, the map $\EM=(H,J)$ is often called \emph{energy-momentum map}.

Consider a closed path $\Gamma$ in the set $\mathcal R$ of regular values of $\EM$ and the $\mathbb T^2$-bundle $\EM^{-1}(\Gamma) \overset{\EM}{\to} \Gamma$. The monodromy of the $\mathbb T^2$-bundle is an automorphism of $H_1(\EM^{-1}(v)) \simeq \mathbb Z^2$ for any $v$ in the image of $\Gamma$. Fixing a basis of $H_1(\EM^{-1}(v))$, monodromy is then characterized by a matrix $M \in \mathrm{SL}(2, \mathbb Z)$. On each fiber $\EM^{-1}(v)$, $v \in \mathcal R$, the existence of the $\mathbb S^1$-action gives a globally defined generator $\gamma_J$ of $H_1(\EM^{-1}(v))$. In a basis $\{\gamma_J,\gamma\}$ of $H_1(\EM^{-1}(v))$ the monodromy matrix has the form
\begin{align*}
  M =
  \begin{pmatrix}
    1 & k \\ 0 & 1
  \end{pmatrix}, \quad k \in \mathbb Z.
\end{align*}
The number $k$ is called the \emph{monodromy number} and completely determines the topology of the $\mathbb T^2$-bundle $\EM^{-1}(\Gamma) \overset{\EM}{\to} \Gamma$. Therefore, the computation of monodromy boils down to the computation of the value of $k$.

In this section we review the computation of monodromy through the variation of the rotation number. We first recall the definition of the rotation number and how it can be used to construct local action coordinates.

\subsection{Rotation Number and its Variation}

In our setting, the momentum $J$ of the $\setS^1$-action can be taken as an action coordinate $I_1$ for the system. A second action coordinate can be constructed in the following way. Consider a point $p$ in a regular fiber $\EM^{-1}(v) \simeq \setT^2$. Furthermore, let $\mathbb S^1\,p$ be the closed orbit of $X_J$ going through $p$. The orbit $\gamma_H(p)$ of $X_H$ starting at $p$ will cross again $\mathbb S^1 p$ at a point $p'$ after a time $T(p)$, called \emph{first return time}, giving a smooth function $T: F^{-1}(\mathcal R) \to \mathbb R$.

\begin{definition}\label{rotnum} The \emph{rotation number} $\Theta(p)$ is the minimal positive time it takes to flow with $X_J$ from $p$ to $p'$.
\end{definition}

The rotation number is a function defined in $\mathcal M$ and taking values in $[0,2\pi)$. With our definition, $\Theta$ is smooth outside its zero level-set $\mathcal Z = \{\, p \in \mathcal M \,|\, \Theta(p) = 0 \,\}$ but is possibly discontinuous at $\mathcal Z$. The set $\mathcal Z$ is typically a union of codimension-1 surfaces in $\mathcal M$ and the function $\Theta$ can possibly tend smoothly to zero from one side and smoothly to $2\pi$ from the other.

Both, the first return time and the rotation number, are invariant under the flows of $X_J$ and $X_H$, and hence are constant on the connected components of the level sets of $\EM$. It follows that they are the pull-back via $\EM$ of functions defined on $\mathcal R$. With a little abuse of notation we will denote the rotation number and the first return time with the same name may they be defined in $\mathcal M$ or in $\mathcal R$. The vector field defined as
\begin{align}
  X_{I_2}=\frac{1}{2\pi} (- \Theta X_J + T X_H)
\end{align}
can be shown to be Hamiltonian and $2\pi$-periodic \cite{Duistermaat1980, Cushman2015}. It is hence associated to the second action coordinate $I_2$ wherever the function $\Theta$ is smooth, that is, outside the set $F(\mathcal Z) \subset \mathcal R$. Note that one can locally define a smooth action coordinate $I_2$ also at $\mathcal Z$ by adding, in a subset of the local neighborhood, an appropriate integer multiple of $2\pi$ to $\Theta$ so as to obtain a locally smooth function.

One of the most important singularities of the map $\EM$, the focus-focus singularity, consists of an isolated point $\bar p$ that is mapped by $\EM$ onto a point $\bar v$ which is a puncture in $\mathcal R$. In this case the zero-set $\mathcal Z$ of $\Theta$ locally consists of curves converging to $\bar v$, typically spiraling around $\bar v$ \cite{Dullin2004a}. Considering a path $\Gamma$ that surrounds such singular value one can add the jumps of $\Theta$ across such curves and obtain an integer multiple of $2\pi$.

\begin{remark}
  Instead of the rotation number $\Theta$ we could have used the \emph{rotation angle} $\widetilde \Theta$, a circle-valued function obtained by composing $\Theta$ with the projection from $\mathbb R$ to $\mathbb R / 2\pi \mathbb Z$. The map $\widetilde \Theta$ can be shown to be smooth, while $\Theta$ can have first-kind discontinuities with jump equal to $\pm 2\pi$.
\end{remark}

The integer obtained by adding up the discontinuities of $\Theta$ along $\Gamma$, and dividing by $2\pi$, reveals the non-triviality of the $\mathbb T^2$-bundle over $\Gamma$. It is connected to the non-existence of global action coordinates and we call it the variation of $\Theta$ along $\Gamma$. We formalize the notion of the variation of an $\mathbb R$-valued function along $\Gamma$ as follows.

\begin{definition} \label{def/var}
Let $g : \Gamma \simeq \setS^1 \to \setR$ be a function with a finite number of discontinuities $p_1,...,p_k \in \Gamma$, whose jumps across the discontinuities are respectively the real numbers 
\begin{align*}
  d_j = \lim_{\varepsilon\to 0^+} \left(g(p_j + \varepsilon)-  g(p_j - \varepsilon)\right),\quad j=1,\dots,k.
\end{align*}
The \emph{variation of $g$ along $\Gamma$} is then defined as
\begin{align*}
  \var_\Gamma g = - \sum_j d_j.
\end{align*}
\end{definition}

\begin{example}
  Consider the function $g : \mathbb S^1 \to \mathbb R$ given by $g(\theta) = \pi + \theta \pmod{2\pi}$ with $\theta \in [0,2\pi)$ parameterizing $\mathbb S^1$. Then $g$ is discontinuous at $\theta = \pi$ and the discontinuity jump is $-2\pi$. Therefore $\var_{\mathbb S^1} g = 2\pi$.
\end{example}

\begin{example}\label{ex:lc}
  Consider any step function $g : \mathbb S^1 \to \mathbb R$. Then the discontinuity jumps must cancel so that $g(0) = \lim_{\theta \to 2\pi^-} g(\theta)$, assuming that $g$ is continuous at $0$. Therefore $\var_{\mathbb S^1} g = 0$.
\end{example}

Note that $X_{I_2} + k X_{I_1}$, $k \in \mathbb Z$, also represents a periodic Hamiltonian vector field associated to the second action coordinate $I_2 + k I_1$. Therefore, a variation of the rotation number by $-2 k \pi$ over $\Gamma$ implies a change of the corresponding action vector field by $k I_1$. Furthermore, since action vector fields generate a basis of the homology group $H_1(\EM^{-1}(v))$ we conclude that, going along $\Gamma$, an initial cycle $\gamma_2$ generated by $X_{I_2}$ is transported to the final cycle $\gamma_2 + k \gamma_J$ and therefore we have a non-trivial monodromy matrix. This comparison shows that the variation 
\begin{equation}\label{varTheta}
  \var_\Gamma \Theta = -2 k \pi
\end{equation}
measures the monodromy number $k$ and hence the non-triviality of the $\mathbb T^2$-bundle over $\Gamma$.

\begin{remark}
  Another way to obtain such integer is to consider the function $\widetilde \Theta|_\Gamma: \Gamma \to \mathbb S^1$ that, being a map from a circle to itself, can possibly have a non-zero degree which is precisely the variation of $\Theta$ along $\Gamma$.
\end{remark}

\subsection{Rotation 1-forms}

In applications, the rotation number $\Theta$ and its variation $\var_\Gamma \Theta$ are typically computed by integrating a closed $1$-form $\vartheta$ along the orbit $\gamma_H$. We formalize here this approach and clarify certain technical aspects of this computation.

\begin{definition}\label{defrotform}
  A \emph{rotation 1-form} is a $1$-form $\vartheta$, defined in an $\mathbb S^1$ invariant subset $\mathcal M'$ of $\mathcal M$, such that $\vartheta$ is closed and $\vartheta(X_J) = 1$. Points in the set $\Pi = \mathcal M \setminus \mathcal M'$ are called \emph{poles} and $\Pi$ is called the \emph{polar set} of $\vartheta$.
\end{definition}

The condition $\vartheta(X_J) = 1$ ensures that the integral of $\vartheta$ measures the natural time along the flow of $X_J$ when integrated along its orbits. Hence, one can use it to define a local angle coordinate along the orbits of $X_J$. Moreover, the conditions in the definition imply that $L_{X_J} \vartheta = d(\vartheta(X_J)) + d\vartheta(X_J,-) = 0$. The latter relation ensures that $\vartheta(X_H)$ is an $X_J$-invariant function, and therefore it descends to a function in the reduced space $\widetilde{\mathcal M} = \mathcal M' / \mathbb S^1$.

The polar set $\Pi$ plays a central role in this work. For this reason we give the geometric intuition for the necessity of introducing $\Pi$ and discuss its role and its properties. We first prove the following result partially characterizing $\mathcal M'$.

\begin{lemma}\label{tpb}
  Let $X_J$ be the generator of an $\mathbb S^1$ action which is free outside fixed points and denote by $\mathcal M_0$ the set of fixed points of the action. Consider the principal circle bundle defined by the flow of $X_J$ on $\mathcal M \setminus \mathcal M_0$. Then $\Pi$ is such that the restriction of the circle bundle to $\mathcal M' = \mathcal M \setminus \Pi$ defines a trivial principal circle bundle. Moreover, $\vartheta$ is a connection $1$-form for the trivial circle bundle defined in $\mathcal M'$.
\end{lemma}

\begin{proof}
  Since the rotation $1$-form $\vartheta$ satisfies $\vartheta(X_J)=1$ and $L_{X_J} \vartheta = 0$ it is a connection $1$-form for the principal circle bundle defined by the flow of $X_J$ on $\mathcal M'$. Moreover, the condition $d\vartheta = 0$ implies that the curvature $2$-form for the corresponding circle bundle is trivial and ensures the triviality of the bundle. Therefore, a rotation 1-form $\vartheta$ can be only defined on a set $\mathcal M'$ so that the restriction of the principal circle bundle to $\mathcal M'$ gives a trivial bundle.
\end{proof}

\begin{lemma}\label{poles-variety}
  If $\bar p$ is a fixed point of the $\mathbb S^1$ action induced by $X_J$ then a rotation $1$-form $\vartheta$ defined in a neighbourhood $U$ of $\bar p$ must have a non-empty polar set $\Pi$ with $\bar p \in \Pi$. Moreover, if the $\setS^1$ action is free in $U \setminus \{ \bar p \}$ then $\Pi \cap U$ must contain a two-dimensional manifold.
\end{lemma}

\begin{proof}
  The rotation 1-form $\vartheta$ cannot be defined at $\bar p$ since $X_J(\bar p) = 0$ but $\vartheta(X_J)(p) = 1$ whenever $\vartheta$ is defined. Therefore $\bar p \in \Pi$. Let now $p$ be a point in an $\mathbb S^1$-invariant open ball $B \ni \bar p$ at which $\vartheta$ is defined. By invariance under the flow of $X_J$, the form $\vartheta$ is defined in all points of the orbit $\mathbb S^1\,p$ through $p$. Since $\vartheta(X_J) = 1$ we have that $\int_{\mathbb S^1\,p} \vartheta = 2\pi$. If $\pi_1(B \setminus \Pi)$ were trivial, then there would exist a disk $\Delta$ in $B \setminus \Pi$, bounded by the orbit $\mathbb S^1\, p$, and then we would get the contradiction $2\pi = \int_{\mathbb S^1\,p} \vartheta = \int_{\Delta} d \vartheta = 0$. Therefore, $\pi_1(B \setminus \Pi)$ must be non-trivial and hence $\Pi$ must contain a nonempty manifold passing through $\bar p$. The non-triviality of $\pi_1(B \setminus \Pi)$ excludes simple possibilities of $\Pi$ being contained in a manifold of dimension $0$ and $1$.
\end{proof}

Note that Lemma~\ref{poles-variety} does not exclude the possibility that $\Pi$ contains a manifold of dimension $3$.

We now consider under what conditions a rotation $1$-form can be defined and how it can be constructed. We start with the following result.

  \begin{lemma}\label{lem-rot-0}
    Suppose that the flow of $X_J$ defines a \emph{trivial} principal circle bundle in $\mathcal M'$. Then there exists a rotation $1$-form $\vartheta$ without poles in $\mathcal M'$.
\end{lemma}

\begin{proof}
  Let $s : \mathcal M' / \mathbb S^1 \to \mathcal M'$ be a smooth section for the bundle. Define an angle $u$ in $\mathcal M'$ as the time it takes for the flow of $X_J$ to move from the image of the section $s$ to a point $p$. Then the $1$-form $\vartheta = du$ can be shown to satisfy the requirements of Definition~\ref{defrotform}. The triviality of the principal bundle ensures that this $1$-form is well-defined and has no poles.
\end{proof}

Then the idea for constructing a rotation $1$-form is that given an $\mathbb S^1$ action in $\mathcal M$ we can obtain a trivial principal circle bundle by taking out a large enough set (which includes the points with non-trivial isotropy) so that in the remaining part we have a trivial principal circle bundle. Moreover, we have the following result.

\begin{lemma}\label{lem-rot}
  Given $\mathcal M$ as above, there always exists a set $\Pi$, finite union of submanifolds of $M$ of codimension at least $1$, outside of which a rotation $1$-form exists.
\end{lemma}

\begin{proof}
  From the theory of CW complexes, one can always assume that $\mathcal M$ admits a stratification of submanifolds of different codimension and a unique (contractible, open, and dense) cell of maximal dimension. The rotation $1$-form can always be defined on the cell of maximal dimension since the corresponding principal circle bundle is trivial.
\end{proof}

\begin{example}
Let $\bar p$ be a fixed point of the circle action possibly defined in a neighborhood of $\bar p$ and free except at $\bar p$. Then the circle action can be locally linearized as $(z,w) \mapsto (e^{it}z,e^{\pm it}w)$, where $(z,w) \in \mathbb C^2$. In particular, the resulting principal circle bundle is isomorphic (up to orientation) to the Hopf fibration and is therefore non-trivial. This means that a rotation $1$-form defined in a punctured neighborhood of a fixed point of the circle action must necessarily have a non-empty set of poles $\Pi$ and the latter should have dimension at least $2$. Assume that we take away the plane $\Pi = \{ z=0 \}$. A bundle section is given by
\begin{align*}
  (\rho_1:=|z|^2,\rho_2:=|w|^2,\chi+i\psi:=\bar{z}w) \mapsto (z,w) = (\sqrt{\rho_1},(\chi+i\psi)/\sqrt{\rho_1}).
\end{align*}
Then $u = \Arg(z)$ and $\vartheta = du = \Im(\bar{z}\,dz/|z|^2)$, or
\begin{align*}
  \vartheta = \frac{x \, dy - y \, dx}{x^2 + y^2},
\end{align*}
where $z = x + i y$.
\end{example}

In what follows we assume that $\Pi$ is a two-dimensional manifold, which is smooth outside fixed points of the $\mathbb S^1$ action. This is a reasonable assumption given that the polar set $\Pi$ of the rotation 1-form $\vartheta$ is a two-dimensional smooth manifold in all examples known to the authors, cf.~Section~\ref{sec/examples}. 

We make use of the following transversality notion.

\begin{definition}\label{deftrans}
  A rotation $1$-form $\vartheta$ with a two-dimensional manifold of poles $\Pi$ is \emph{transversal to $\EM$} if $\EM|_\Pi$ has rank $1$ outside fixed points of the $\mathbb S^1$ action.
\end{definition}

Note that the rank of $\EM|_\Pi$ cannot equal $2$ since both $\EM$ and $\Pi$ are $X_J$ invariant.

\begin{lemma}\label{lemma/transversality}
Consider a rotation 1-form $\vartheta$ transversal to $\EM$, let $\Pi_r = \Pi \setminus \{ \text{fixed points of the $\mathbb S^1$ action} \}$, and assume that the circle action is free outside fixed points. Then $\EM(\Pi_r)$ is a smooth open one-dimensional manifold and for each $v \in \EM(\Pi_r)$ the intersection $\EM^{-1}(v) \cap \Pi_r$ consists of a finite number of $\mathbb S^1$ orbits.
\end{lemma}

\begin{proof}
  By our assumptions on $\Pi$, $\Pi_r$ is a smooth two-dimensional manifold. Since the $\mathbb S^1$ action is free on $\Pi_r$, the reduced $\Pi_r / \mathbb S^1$ is a one-dimensional manifold. By transversality to $\EM$, $\EM|_{\Pi_r}$ is of rank $1$, which implies that it reduces to a map $f$ of rank $1$ on $\Pi_r / \mathbb S^1$. Therefore, the map $f$ is an immersion and it follows that its image $f(\Pi_r / \mathbb S^1) = \EM(\Pi_r)$ is smooth one-dimensional. 

  For each $v \in F(\Pi_r)$, by transversality, $\EM^{-1}(v) \cap \Pi_r$ is one-dimensional. From $\mathbb S^1$-invariance of $F$ and $\Pi_r$, $F^{-1}(v) \cap \Pi_r$ is a union of $\mathbb S^1$ orbits. There is a finite number of them by transversality of $\vartheta$.
\end{proof}

Note that the manifold of poles $\Pi$ is well-defined within the disk $D$ above which a circle action is well-defined. Of course, if as in many examples the action is global, then the manifold of poles $\Pi$ is defined globally.

\subsection{Variation of the Rotation Number and Rotation 1-forms}

\begin{figure}
  \includegraphics[width=0.5\columnwidth]{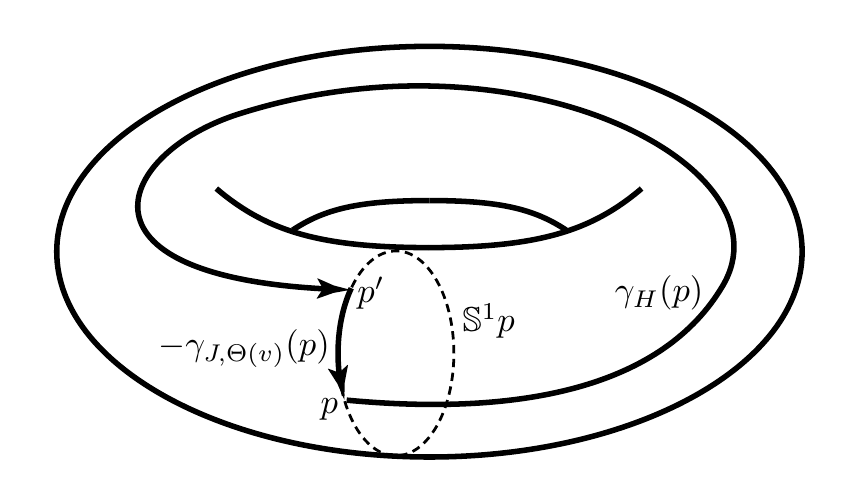}
  \caption{The orbit segment $\gamma_H(p)$ on a torus $\EM^{-1}(v)$. The cycle $\delta_H(p)$ is defined by adding to $\gamma_H(p)$ the curve $-\gamma_{J,\Theta(v)}(p)$ from $p'$ to $p$ along $\mathbb S^1 p$.}
  \label{torus}
\end{figure}

Let $v \in \mathcal R$, and let $p \in \EM^{-1}(v) \subset \mathcal M$. As in the definition of the rotation number let $\mathbb S^1 p$ be the $X_J$ orbit through $p$ and $\gamma_H(p)$ the segment of the orbit of $X_H$ on $\EM^{-1}(v)$ starting from $p$ and flowing until it meets $\mathbb S^1 p$ at a point $p'$. We call $\delta_H(p)$ the closed curve which is the result of joining $\gamma_H(p)$ with the curve $-\gamma_{J,\Theta(v)}(p)$. The latter is obtained by flowing along $X_J$ from $p'$ for time $-\Theta(v)$, that is, until closing at $p$.

The integration of $\vartheta$ along the paths $\gamma_H(p)$ and $\delta_H(p)$ does not depend on the choice of the point $p$ in the fiber $\EM^{-1}(v)$. Therefore, for $v \in \mathcal R$ we define
\begin{align}\label{def-Phi}
  \Phi(v) = \int_{\gamma_H(p)} \vartheta,
\end{align}
where $p$ is any point in $\EM^{-1}(v)$. Note that $\Phi$ is not defined and may not be extended by continuity whenever $\EM^{-1}(v)$ intersects $\Pi$.
 
\begin{lemma}\label{rotation angle and rotation 1-form integral} The following facts hold:
\begin{enumerate}[label={(\alph*)}]
\item Let $\Gamma$ be a closed path in the set of regular values $\mathcal R$ which transversally intersects $\EM(\Pi)$. Then
  \begin{align*}
    \var_\Gamma \Phi = \var_\Gamma \Theta,
  \end{align*}
  where $\Theta$ is the rotation number and $\Phi$ is given by Eq.~\eqref{def-Phi}.
\item Let $v \in \mathcal R$. Then $\Phi(v) = \Theta(v) \pmod{2\pi}$ if and only if there
  exists a cycle $\delta \in H_1(\EM^{-1}(v))$, independent of the cycle defined by the
  $X_J$-orbit, such that $\int_\delta \vartheta = 0$.
\end{enumerate}
\end{lemma}

\begin{proof}
  (a) Consider the representative of $\delta_H$ that goes from $p$ to $p'$ along $\gamma_H$ and then from $p'$ to $p$ along the flow of $X_J$ for time $-\Theta$. By construction, such a path is $\gamma_H - \gamma_{J,\Theta}$ where $\gamma_{J,\Theta}$ is the time-$\Theta$ orbit of $X_J$ from $p$ to $p'$. Then
\begin{align*}
  \int_{\delta_H} \vartheta
  = \int_{\gamma_H} \vartheta -  \int_{\gamma_J,\Theta} \vartheta 
  = \Phi -  \Theta ,
\end{align*}
where we used that $\Theta = \int_{\gamma_{J,\Theta}} \vartheta$, since $\vartheta(X_J)=1$.
Therefore,
\begin{equation}\label{thph}
  \Theta = \Phi - \int_{\delta_H} \vartheta.
\end{equation}

Parameterize $\Gamma$ by $\Gamma : [0,2\pi]\to \mathcal R : s \mapsto \Gamma(s)$. The function $s \mapsto \int_{\delta_H(p(s))} \vartheta$ is locally constant along $\Gamma$ since $\vartheta$ is a closed 1-form and the initial points $p(s)$ for the construction of the cycles $\delta_H(p(s))$ can be chosen so that these cycles form a cylinder. When $\delta_H(p(s))$ meets a pole of $\vartheta$ then $\int_{\delta_H(p(s))} \vartheta$ is not defined and the function $s \mapsto \int_{\delta_H(p(s))} \vartheta$ has a discontinuity which, because of the local constancy, must be a jump discontinuity. Therefore, $s \mapsto \int_{\delta_H(p(s))} \vartheta$ is a step function. The rotation number $\Theta$ also only has jump discontinuities and these two facts, together with Eq.~\eqref{thph}, imply that $\Phi(\Gamma(s))$ also has only jump discontinuities along $\Gamma$. Therefore, using that all functions involved only have jump discontinuities, we obtain
\begin{align*}
  \var_\Gamma \Theta = \var_\Gamma \Phi - \var_\Gamma \int_{\delta_H} \vartheta.
\end{align*}
Since $\int_{\delta_H} \vartheta$ is a step function we have $\var_\Gamma \int_{\delta_H} \vartheta = 0$, cf.~Example~\ref{ex:lc}. Therefore,
\begin{align*}
  \var_\Gamma \Theta = \var_\Gamma \Phi.
\end{align*}

(b) Suppose that there exists a cycle $\delta$ which is independent of $\gamma_{J,2\pi}$ and satisfies $\int_\delta \vartheta = 0$. The cycle $\delta_H$ can be written as $\delta_H = k_1 \delta + k_2 \gamma_{J,2\pi}$ with $k_1, k_2 \in \mathbb Z$ and $k_1 \ne 0$. Therefore,
\begin{equation*}
  \int_{\delta_H} \vartheta = k_2 \int_{\gamma_{J,2\pi}} \vartheta = 2 \pi k_2,
\end{equation*}
and Eq.~\eqref{thph} gives $\Theta - \Phi = 0 \pmod{2\pi}$. In the opposite direction, we have that $\Theta - \Phi = 0 \pmod{2\pi}$ implies, using Eq.~\eqref{thph}, that $\int_{\delta_H} \vartheta = 2 \pi k$ for some $k \in \mathbb Z$. Then the cycle $\delta = \delta_H - k \gamma_{J,2\pi}$ satisfies $\int_\delta \vartheta = 0$.
\end{proof}

\begin{remark}
  There may be fibers $\EM^{-1}(v)$ for $v \in \mathcal R$ such that $\EM^{-1}(v) \cap \Pi \ne \emptyset$. Thus $\Phi = \int_{\gamma_H} \vartheta$ and $\int_{\delta_H} \vartheta$ are not defined on these fibers. Nevertheless, $\Theta$ is always defined (by construction) and, therefore, the difference $\Phi - \int_{\delta_H} \vartheta$ extends to a well-defined function on such fibers.
\end{remark}

\begin{remark}
  Lemma~\ref{rotation angle and rotation 1-form integral} shows that $\var_\Gamma \Phi$ is independent of the choice of the rotation 1-form $\vartheta$ and always equals $\var_\Gamma \Theta$. This means that we can choose $\vartheta$ in such a way so as to simplify the computation of the variation, even if it does not give the correct value for the rotation number $\Theta$ on each fiber.
\end{remark}

\subsection{Proof of the Main Theorem~\ref{thm/main}}
\label{sec/main-proof}

\begin{figure}
\includegraphics[width=0.75\columnwidth]{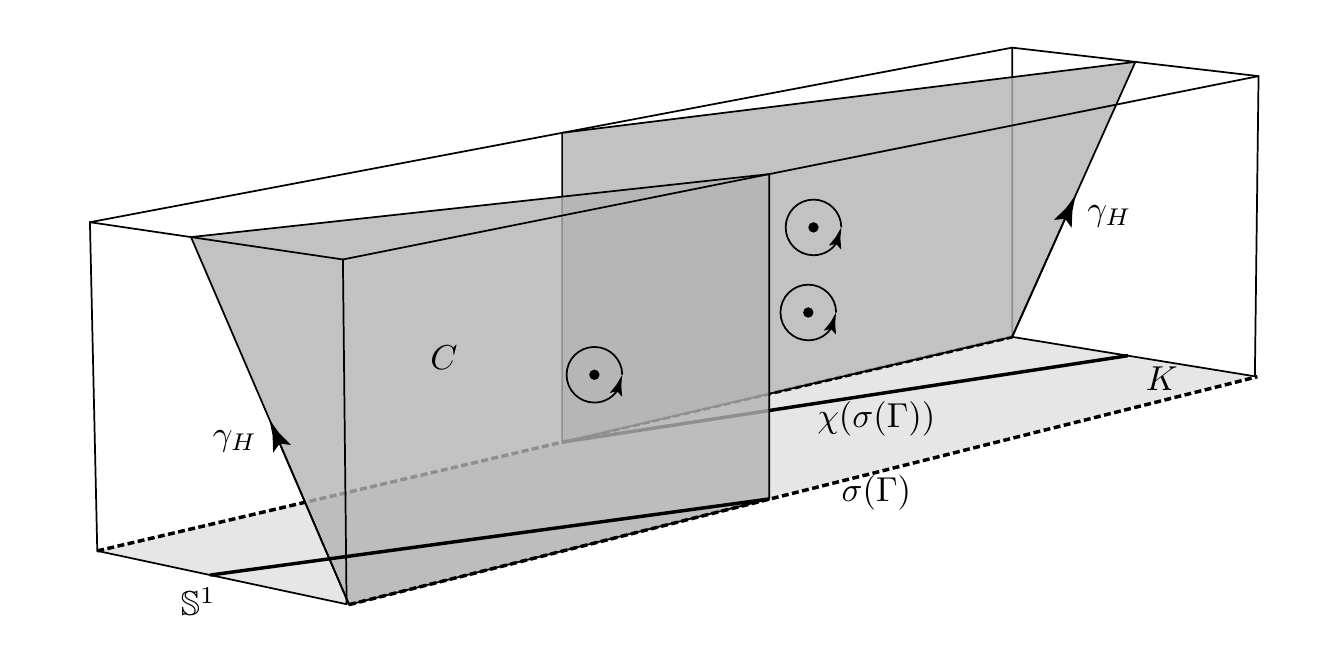}
\caption{The fibration $\EM$ above $\Gamma$ using the same representation as in Figure~\ref{fibration}. The torus $K$ is represented by the lower, light gray, face and consequently also by the opposite upper face. The cylinder $C$ of $X_H$ orbits starting at $\sigma(\Gamma)$ is represented by the dark gray surface. The dashed lines represent the section $\sigma(\Gamma)$. The upper side of $C$ is also drawn with a thicker line on the lower face. Note that the lines marked by $\gamma_H$ represent the same $X_H$ orbit on $\EM^{-1}(\Gamma(0)) = \EM^{-1}(\Gamma(1))$. $C$ intersects the polar set $\Pi$ at a finite number of isolated points $p_{ij}$, cf.~Figure~\ref{fibration}. The cycles $\delta_{ij}$ around $p_{ij}$ are defined on $C$.}
\label{fig/main-thm-proof}
\end{figure}

To prove Theorem~\ref{thm/main} consider a section $\sigma : \Gamma \to \EM^{-1}(\Gamma)$ of the $\mathbb T^2$ bundle over $\Gamma$ and the $2$-torus $K = \{ \mathbb S^1 \sigma(v) : v \in \Gamma \}$. Then consider the cylinder $C$ made up of orbit segments $\gamma_H(\sigma(v))$ of $X_H$ as $v$ moves along $\Gamma$, see Figure~\ref{fig/main-thm-proof}. Specifically, the orbit segment $\gamma_H(\sigma(v))$ of $X_H$ starts at $p = \sigma(v)$ and ends at the first point $p' \in \mathbb S^1 p \subset K$ where the orbit intersects $\mathbb S^1 p$. This construction defines the map $\chi : \sigma(\Gamma) \to K$ sending $p$ to $p'$. In terms of homology classes in $H_1(K,\mathbb Z)$ we have $\chi(\sigma(\Gamma)) = \sigma(\Gamma) + \ell \, [\mathbb S^1]$ for some $\ell \in \mathbb Z$. Here $[\mathbb S^1]$ is the homology class represented by any $X_J$ orbit of period $2\pi$.

Parameterize $\Gamma$ by $s \in [0,1]$ with $s$ increasing along the traversing direction of $\Gamma$ and parameterize each orbit segment $\gamma_H(\sigma(v))$ by $t \in [0,1]$ with $t$ increasing along the flow of $X_H$. Then $C$ is parameterized by $(s,t) \in [0,1] \times [0,1]$ and such choice fixes an orientation on $C$.

Since $\vartheta$ is closed, Stokes' theorem gives that $\int_{\partial C} \vartheta$ equals the sum of the integral of $\vartheta$ along positively oriented cycles $\delta_{ij}$ encircling the poles of $\vartheta$ on $C$,
\begin{align*}
  \int_{\partial C} \vartheta = \sum_{ij} \int_{\delta_{ij}} \vartheta.
\end{align*}
The boundary of $C$ is $\partial C = - \chi(\sigma(\Gamma)) + \sigma(\Gamma)$, therefore
\begin{align*}
  \int_{\partial C} \vartheta = - \int_{\chi(\sigma(\Gamma))} \vartheta + \int_{\sigma(\Gamma)} \vartheta = - \ell \int_{\mathbb S^1} \vartheta = - 2 \ell \pi.
\end{align*}
Moreover, the variation of $\Phi$ along $\Gamma$ is given by 
\begin{align*}
  \var_\Gamma \Phi = - \sum_{ij} \int_{\delta_{ij}} \vartheta,
\end{align*}
giving
\begin{align*}
  \var_\Gamma \Phi = 2 \ell \pi .
\end{align*}
Recall from Equation~\eqref{varTheta} that the monodromy number $k$ equals $- \tfrac{1}{2\pi} \var_\Gamma \Theta $. By Lemma \ref{rotation angle and rotation 1-form integral}, we then have that 
\begin{align*}
  k = - \frac{1}{2\pi} \var_\Gamma \Theta = - \frac{1}{2\pi} \var_\Gamma \Phi = -\ell = \frac{1}{2\pi} \sum_{ij} \int_{\delta_{ij}} \vartheta.
\end{align*}
This concludes the proof of Theorem~\ref{thm/main}.

\section{Examples}
\label{sec/examples}

In this section we apply the concepts introduced in Section~\ref{sec/rotation} to three specific examples: the champagne bottle, the spherical pendulum, and a system on the symplectic manifold $\mathbb S^2 \times \mathbb S^2$.

\subsection{The Champagne Bottle}
\label{sec/champagne}

The champagne bottle consists of a particle in the plane $\mathbb R^2$ which moves under the influence of a conservative force whose potential energy is
\begin{align*}
  V(q_1,q_2) = (q_1^2 + q_2^2)^2 - (q_1^2 + q_2^2).
\end{align*}
The phase space of this system is the cotangent bundle of $\mathbb R^2$, diffeomorphic to $\mathbb R^4$, with the canonical symplectic structure $\omega = dp_1\wedge dq_1 + dp_2 \wedge dq_2$. The Hamiltonian function is $H(q,p) = \frac 12 (p_1^2+p_2^2) + V(q_1,q_2)$.

This system admits the integral of motion $J = q_1 p_2 - q_2 p_1$. The function $J$ is the momentum of the $1:(-1)$ oscillator that rotates clockwise in the $(q_1,q_2)$-plane and counterclockwise in the $(p_1,p_2)$-plane. Its infinitesimal action is the vector field $X_J = q_2 \partial_{q_1} - q_1 \partial_{q_2} - p_2 \partial_{p_1} + p_1 \partial_{p_2}$.

This system admits the global rotation 1-form
\[
\vartheta = \frac{q_1 dq_2 - q_2 d q_1}{q_1^2+q_2^2},
\]
whose polar set $\Pi$ is the plane $q_1 = q_2 = 0$. The transversality condition of Definition \ref{deftrans} is easily verified. In fact, the intersection of $\Pi$ with the critical fiber is only the critical point $(0,0,0,0)$, and the energy-momentum map $F = (H,J)$ restricted to $\Pi$ is the function $(p_1,p_2) \mapsto ((p_1^2+p_2^2)/2,0)$ which has rank 1 at all points of the plane except the critical point. The projection of $\Pi$ in the energy-momentum domain is the $H$ positive semi-axis, see Figure~\ref{ex}.

\begin{figure}
\centering
\includegraphics[width=0.3\linewidth]{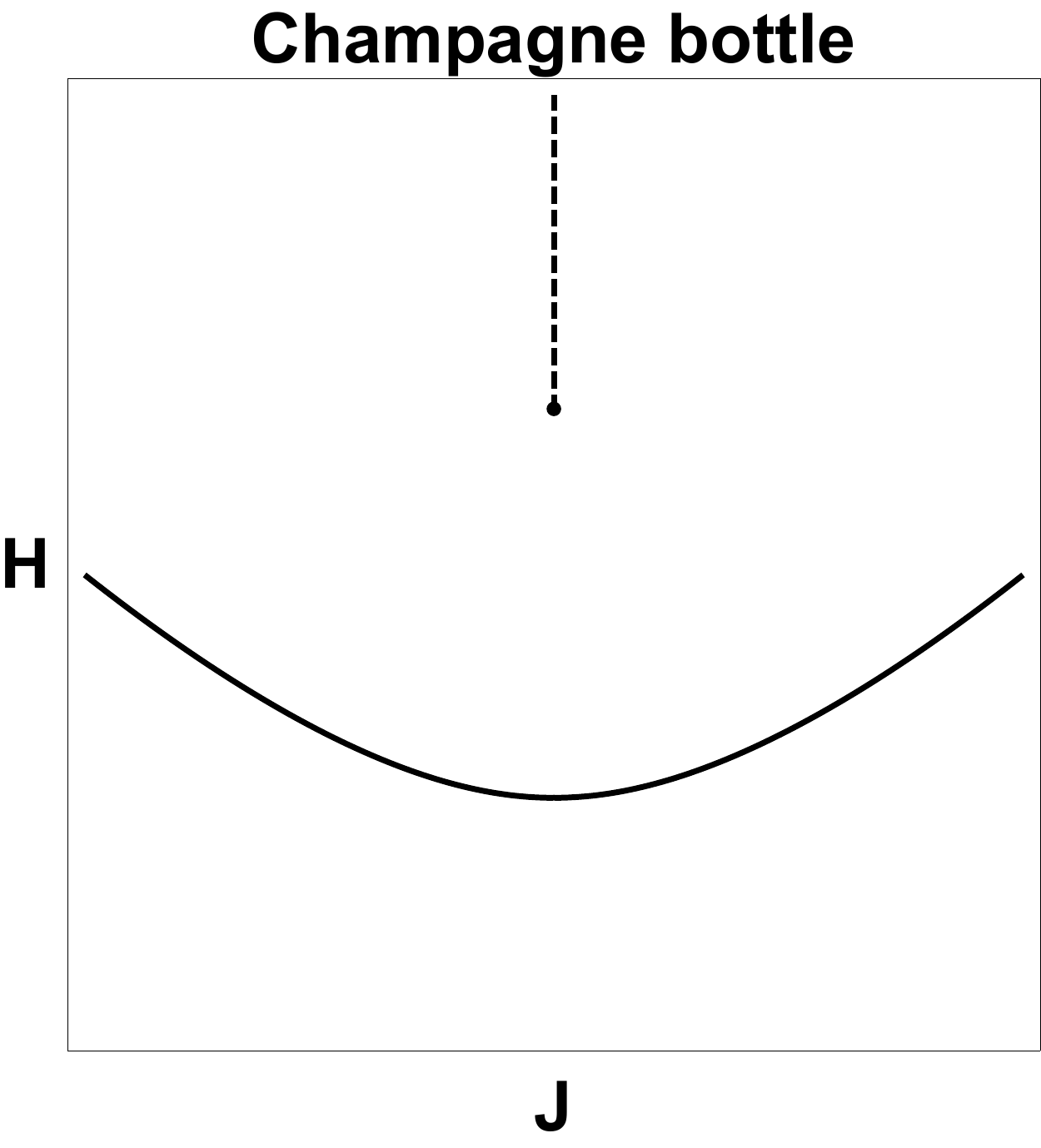}\hfill%
\includegraphics[width=0.3\linewidth]{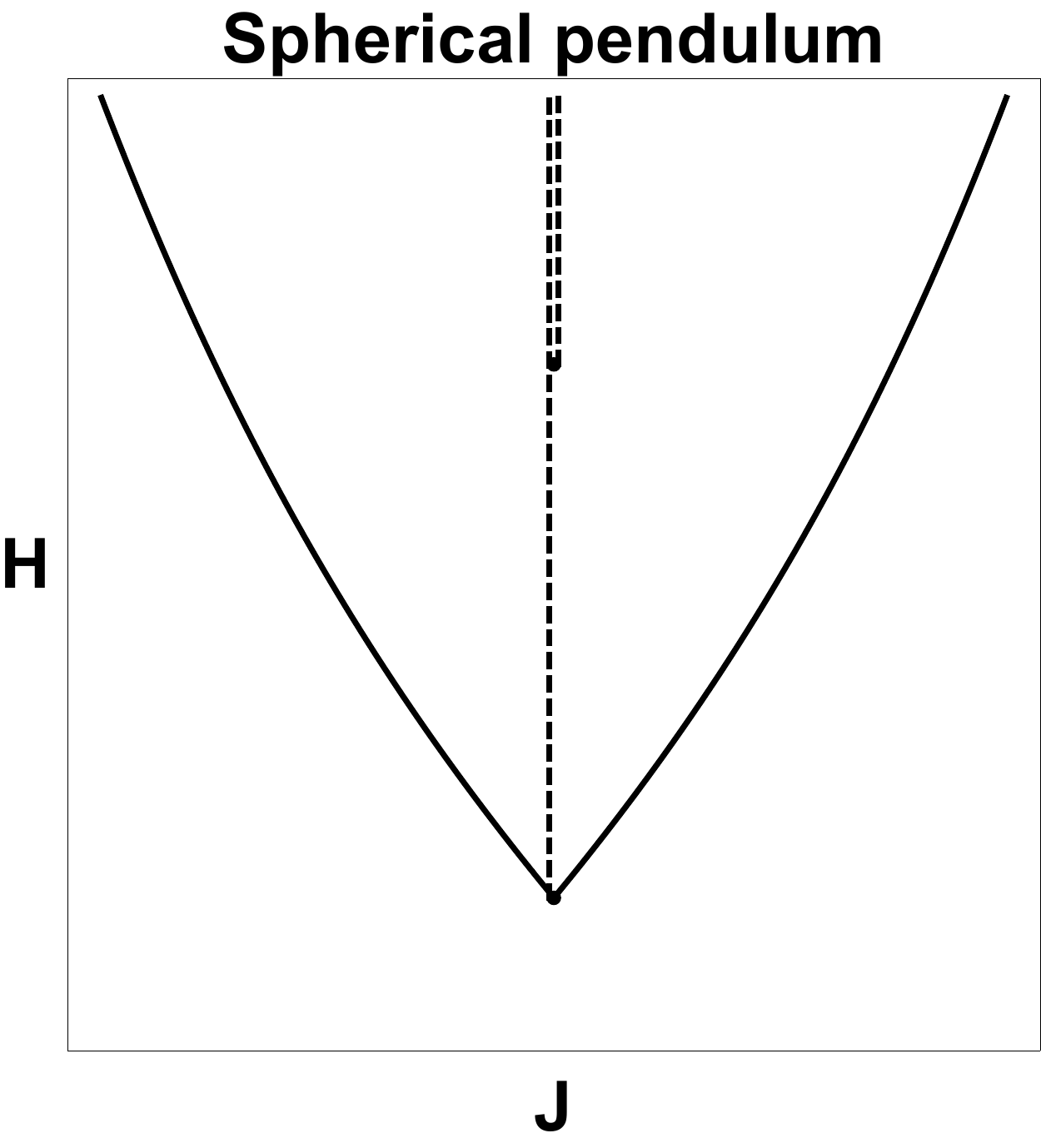}\hfill%
\includegraphics[width=0.3\linewidth]{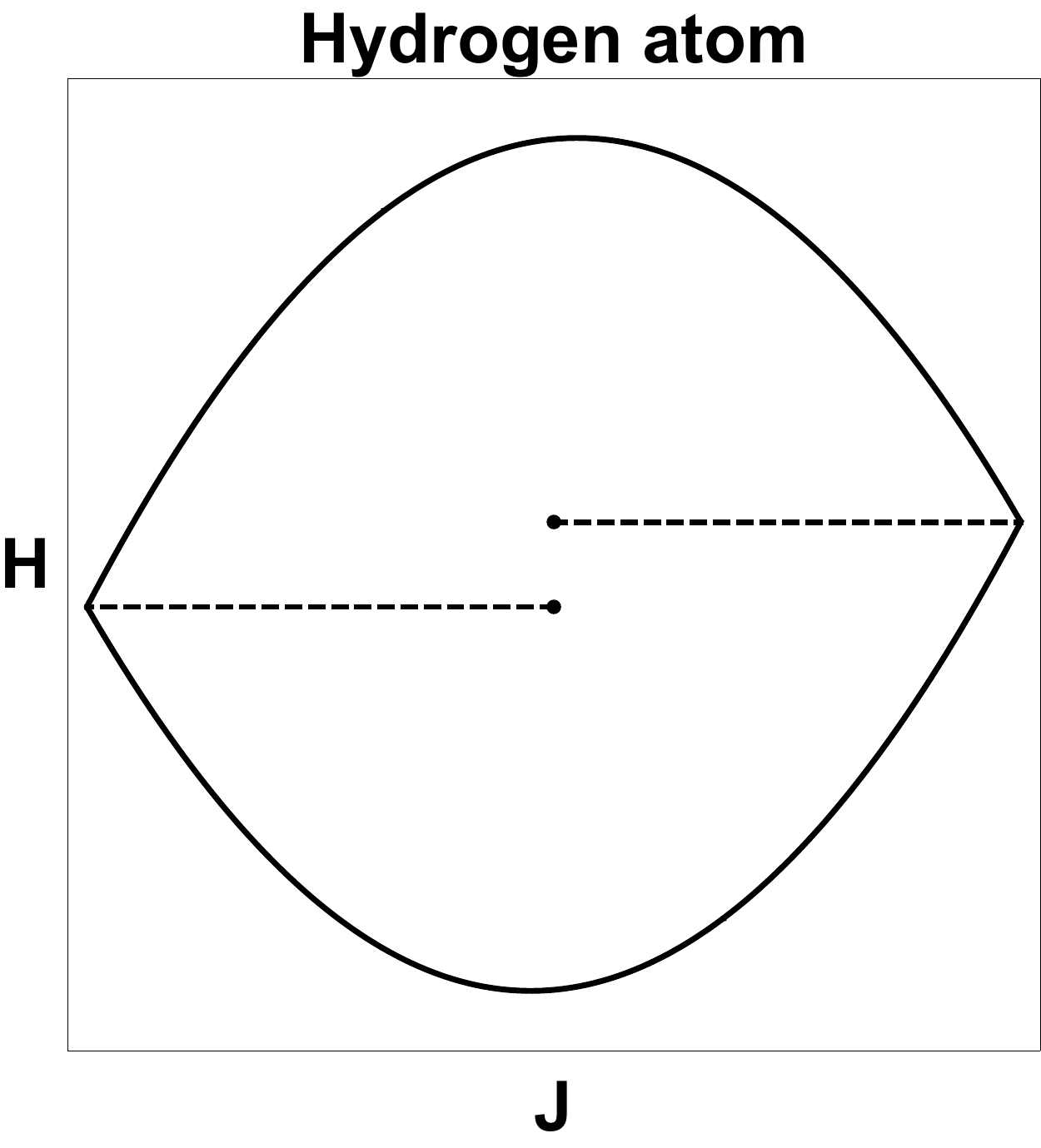}
\caption{The energy-momentum domain for the three examples in
  Section~\ref{sec/examples}. The isolated points correspond to focus-focus
  singularities, the dashed lines are the projection of the domain of $\Pi$, the
  manifold of poles of the chosen rotation 1-form. For the spherical pendulum
  (middle panel) the two lines are both along the $H$-axis; they have been drawn
  slightly shifted to make them both visible.}
\label{ex}
\end{figure}

To compute monodromy by applying Theorem~\ref{thm/main}, consider a closed path $\Gamma$ that encircles the origin in a counterclockwise direction and transversally crosses $\EM(\Pi) = \{(h,j),\,j=0,\,h\ge0\}$ at a point $(\varepsilon,0)$. Then $\Gamma$ is locally parameterized by $(h,j)=(g(s),-s)$ with $g(0)=\varepsilon$. The polar orbit $P = \EM^{-1}(\varepsilon,0) \cap \Pi$ is given by $p_1^2+p_2^2 = 2\varepsilon$ and $q_1^2 + q_2^2 = 0$. A tubular neighborhood $U$ of $P$, contained in $\EM^{-1}(\Gamma)$, admits a chart $(q_1,q_2,\theta)$ where $\theta = \arg(p_1 + i p_2) \in \mathbb S^1$ while $(q_1,q_2)$ lie in a small disk $V$ containing the origin in $\mathbb R^2$. Taking $V$ (and subsequently $U$) sufficiently small and using the transversality condition, the cylinder $C$ of orbits of $X_H$ defined in the proof of Theorem~\ref{thm/main}, Section~\ref{sec/main-proof}, intersects $U$ along a disk that can be parameterized by $(q_1,q_2)$ and $P$ is represented by $q_1=q_2=0$. We have to determine the orientation of the chart $(q_1,q_2)$ with respect to the orientation used in the proof of Theorem~\ref{thm/main}. The latter is defined by $(s,t)$ where $s$ is an increasing parameter along $\Gamma$ and $t$ is time. Then we need to check the determinant
\begin{align*}
  D = \begin{vmatrix}
    \frac{\partial q_1}{\partial s} & \frac{\partial q_1}{\partial s} \\
    \frac{\partial q_1}{\partial t} & \frac{\partial q_1}{\partial t}
  \end{vmatrix}
  = p_2 \frac{\partial q_1}{\partial s} - p_1 \frac{\partial q_1}{\partial s},
\end{align*}
which we can evaluate at $q_1=q_2=0$. Since $\partial q_j / \partial t = \dot{q}_j = p_j$, we find
\begin{align*}
  -1 = \frac{dj}{ds} = D + \left( q_1 \frac{\partial p_2}{\partial s} - q_2 \frac{\partial p_1}{\partial s} \right).
\end{align*}
Evaluating the last relation at $q_1=q_2=0$, gives $D=-1$. This implies that a cycle $\delta$, which is positively oriented on $C$, is negatively oriented in the $(q_1,q_2)$-plane and therefore
\begin{align*}
  k = \frac{1}{2\pi} \int_\delta \vartheta = - 1.
\end{align*}

\subsection{The Spherical Pendulum}
\label{sec/spherical-pendulum}

The spherical pendulum is a Hamiltonian system in the cotangent bundle of the sphere $T^*\mathbb S^2$. This manifold can be symplectically embedded in the cotangent bundle of $\mathbb R^3$, that is diffeomorphic to $\mathbb R^6$ with canonical coordinates $q_i, p_i$, $i=1,2,3$. In these coordinates the Hamiltonian of the system is the restriction to $T^*S^2 = \{ (q,p) \,|\, \|q\|^2 = 1, q\cdot p = 0\}$ of the function $H(q,p) = \frac 12 |p|^2 + q_3$. This Hamiltonian commutes with the function $J = q_1 p_2 - q_2 p_1$.
This system admits the global rotation 1-form
\begin{align*}
\vartheta = \frac{q_1 dq_2 - q_2 d q_1}{q_1^2+q_2^2}.
\end{align*}

In this case the poles of the rotation 1-form are the points satisfying the two equations $q_1 = q_2 = 0$, that form the two planes
\begin{align*}
\Pi_\pm =\{ (0,0,\pm 1,p_1,p_2,0) \,|\, p_1,p_2 \in \mathbb R \}.
\end{align*}
The restriction of the energy-momentum map $\EM=(H,J)$ to the two planes is the function $(p_1,p_2) \mapsto (\tfrac12(p_1^2+ p_2^2) \pm 1, 0)$, which has rank 1 at all points except the poles (defined by $p_1 = p_2 = 0$), which are singular points for the system. The image of this map, that is $\EM(\Pi)$, consists of two rays, subsets of the $H$-axis (see Figure~\ref{ex}).

To compute monodromy in this example we follow the same argument as for the champagne bottle, Section~\ref{sec/champagne}. The main difference is that now $\Gamma$ intersects $\EM(\Pi)$ at two distinct points and $\EM^{-1}(\Gamma)$ contains three polar orbits. Consider a closed path $\Gamma$ that encircles the focus-focus value $(h,j)=(1,0)$ in a counterclockwise direction and transversally crosses $\EM(\Pi) = \{(h,j),\,j=0,\,h\ge-1\}$ at the points $(1\pm\varepsilon,0)$, $\varepsilon > 0$. When $\Gamma$ crosses $\EM(\Pi)$ at $(1+\varepsilon,0)$ it is locally parameterized by $(h,j)=(g(s),-s)$ with $g(0)=1+\varepsilon$. There are two polar orbits $P_\pm$ on $\EM^{-1}(1+\varepsilon,0)$, given by $p_1^2+p_2^2 = 2(1+\varepsilon\mp1)$, $p_3=0$, and $q_1^2 + q_2^2 = 0$, $q_3=\pm1$. Each polar orbit $P_\pm$ admits a tubular neighborhood $U_\pm$, contained in $\EM^{-1}(\Gamma)$. Each $U_\pm$ admits a chart $(q_1,q_2,\theta)$, where $\theta = \arg(p_1 + i p_2) \in \mathbb S^1$. Taking $U_\pm$ sufficiently small and using the transversality condition, the cylinder $C$ of orbits of $X_H$, intersects each of $U_\pm$ along a disk that can be parameterized by $(q_1,q_2)$ and $P_\pm$ is represented by $q_1=q_2=0$. For the orientation we check the determinant
\begin{align*}
  D = \begin{vmatrix}
    \frac{\partial q_1}{\partial s} & \frac{\partial q_1}{\partial s} \\
    \frac{\partial q_1}{\partial t} & \frac{\partial q_1}{\partial t}
  \end{vmatrix}
  = p_2 \frac{\partial q_1}{\partial s} - p_1 \frac{\partial q_1}{\partial s},
\end{align*}
which we can evaluate at $q_1=q_2=0$.
We further have
\begin{align*}
  -1 = \frac{dj}{ds} = D + \left( q_1 \frac{\partial p_2}{\partial s} - q_2 \frac{\partial p_1}{\partial s} \right),
\end{align*}
which, evaluated at $q_1=q_2=0$, gives $D=-1$. This implies that the cycles $\delta_\pm$ should be negatively oriented in the $(q_1,q_2)$-plane and therefore
\begin{align*}
  \int_{\delta_\pm} \vartheta = -2\pi.
\end{align*}
The polar orbit $P_0$ on $\EM^{-1}(1-\varepsilon,0)$ is given by $p_1^2+p_2^2 = 2(2-\varepsilon)$, $p_3=0$, and $q_1^2 + q_2^2 = 0$, $q_3=-1$. Working as above, we find $D = dj/ds = 1$, therefore
\begin{align*}
  \int_{\delta_0} \vartheta = 2\pi.
\end{align*}
In conclusion,
\begin{align*}
  k = \frac{1}{2\pi} \left( \int_{\delta_0} \vartheta + \int_{\delta_+} \vartheta + \int_{\delta_-} \vartheta \right) = - 1.
\end{align*}
Note, in particular, that the polar orbits $P_0$ and $P_-$ that belong to $\EM(\Pi_+)$ are bounded away from the focus-focus point $(0,0,1,0,0,0)$ as $\varepsilon$ goes to zero and their contributions to the monodromy number cancel out. This means that only the polar orbit $P_+$ which approaches the focus-focus point as $\varepsilon$ goes to zero contributes to monodromy. We show that monodromy is locally determined in Section~\ref{sec/local}.

Note that the argument we give here for the spherical pendulum, works in exactly the same way, for more general systems of the form $H(q,p)=\tfrac12|p|^2+V(q_3)$ on $T^*\mathbb S^2$, provided that the path $\Gamma$ lies in the set of regular values of $F$ and it transversally intersects $F(\Pi)$ which is a subset of the $H$-axis. In particular, this includes the case where the system does not have a focus-focus singularity but, instead, a more complicated arrangement of critical values forming an ``island'', see Ref.~\citenum{Efstathiou2005}, Chapter 4.

\subsection{The Hydrogen Atom in Crossed Fields}
\label{sec/hydrogen}

After a first reduction, the hydrogen atom in crossed electric and magnetic fields turns into a Hamiltonian system defined in $\mathbb S^2 \times \mathbb S^2$ that can be embedded into the manifold $\mathbb R^3 \times \mathbb R^3$ endowed with the Poisson structure coming from the Lie algebra $\mathfrak{so}(3)$ (that is $\{x_i,x_j\} = \sum_k \varepsilon_{i,j,k} x_k$ and $\{y_i,y_j\} = \sum_k \varepsilon_{i,j,k} y_k$ where $\varepsilon_{i,j,k}$ is the signature of the permutation $1 \to i$, $2 \to j$ $3\to k$). The phase space $\mathbb S^2 \times \mathbb S^2$ is a symplectic leaf of this space.

The Hamiltonian function for this system is $H = a x_3 + b y_3 + H_2$, with $H_2$ a function of degree two or higher in the variables depending on the parameters $a,b$ (see Ref.~\citenum{Efstathiou2005}, Chapter 3, for a detailed description). This system can be normalized so as to admit an invariance under the Poisson action of $J = x_3 + y_3$, that induces a simultaneous clockwise rotation in the two copies of $\mathbb R^3$ about the $x_3$ and $y_3$ axes respectively. In this case, regardless of the choice of $H_2$, a global rotation 1-form is
\begin{align*}
  \vartheta = \frac{x_1 dx_2 - x_2 d x_1}{x_1^2+x_2^2},
\end{align*}
and its poles correspond to two submanifolds $\Pi_\pm = (0,0,\pm 1) \times \mathbb S^2$.

Theorem \ref{thm/main} applies. We illustrate the application considering a toy model with $H = a x_3 + x_1 y_2 - x_2 y_1$. The restriction of the function $F =(H,J)$ to $\Pi_\pm$ is the function $(\pm a, \pm 1 + y_3)$, which has rank 1 and projects onto two horizontal lines connecting each focus-focus critical value to the elliptic-elliptic critical value at the same height (see Figure~\ref{ex}). {\red The addition of the term $by_3$ to $H$ or the choice of a different $H_2$ term would deform this picture but leave it qualitatively the same (as long as the quadratic part does not cause a bifurcation of the system).

To apply Theorem~\ref{thm/main}, consider a closed path $\Gamma$ that encircles both focus-focus values in a counterclockwise direction and transversally crosses $\EM(\Pi)$ at the points $\pm (a,\varepsilon)$. Near each point $\pm (a,\varepsilon)$ the path  $\Gamma$ is locally parameterized by $(h,j)=(\pm s, g_\pm(s))$ with $g_\pm(0)=\pm \varepsilon$. The polar orbits $P_\pm$ on $\EM^{-1}(\pm(a,\varepsilon))$ are given by $y_1^2+y_2^2 = $ and $x_1^2 + x_2^2 = 0$, $x_3 = \pm1$. Tubular neighborhoods $U_\pm$ of $P_\pm$, contained in $\EM^{-1}(\Gamma)$, admit charts $(x_1,x_2,\theta)$ where $\theta = \arg(y_1 + i y_2) \in \mathbb S^1$. Taking $U_\pm$ sufficiently small and using the transversality condition, the cylinder $C$ of orbits of $X_H$ defined in the proof of Theorem~\ref{thm/main}, Section~\ref{sec/main-proof}, intersects $U_\pm$ along disks that can be parameterized by $(q_1,q_2)$ and $P_\pm$ are represented by $q_1=q_2=0$. We have to determine the orientation of the charts $(q_1,q_2)$ with respect to the orientation used in the proof of Theorem~\ref{thm/main}. Then we need to check the determinant
\begin{align*}
  D = \begin{vmatrix}
    \frac{\partial x_1}{\partial s} & \frac{\partial x_2}{\partial s} \\
    \frac{\partial x_1}{\partial t} & \frac{\partial x_2}{\partial t}
  \end{vmatrix}
  = (\mp y_2) \frac{\partial x_1}{\partial s} - (\mp y_1) \frac{\partial x_2}{\partial s},
\end{align*}
which is being evaluated at $x_1=x_2=0$.
We further have
\begin{align*}
  \pm 1 = \frac{dh}{ds} = \mp D.
\end{align*}
Therefore, in both cases $D=-1$ and we get
\begin{align*}
  k = \frac{1}{2\pi} \left( \int_{\delta_+} \vartheta + \int_{\delta_-} \vartheta \right) = - 2.
\end{align*}

\section{Local Monodromy}
\label{sec/local}

We now concentrate on a point $\bar p \in \mathcal M$ that is a focus-focus singularity \cite{Zung1997, Bolsinov2004}. Such singularities are isolated, rank-zero, singularities of $\EM$. This implies that $X_J(\bar p) = X_H(\bar p) = 0$. Moreover, we assume that the fiber containing $\bar p$ is a singly pinched torus, that is, the only critical point on this fiber is $\bar p$. Under these assumptions we show that the main theorem, Theorem~\ref{thm/main}, can be applied locally to determine the monodromy number $k$ near $\bar v = \EM(\bar p)$. In particular, we show that $k$ can be determined by restricting our attention to a non-saturated neighborhood of $\bar p$. We then make a specific choice of the rotation $1$-form and we use it to compute that $k=-1$.

\subsection{Local Fibration and its Complement}

Let us consider the fibration induced by $F$ in a neighborhood of a focus-focus point $\bar p$ of a 2 degree of freedom Hamiltonian system. Let $\bar v = \EM(\bar p)$ and denote by $\rho$ the reduction map of the Hamiltonian $\mathbb S^1$-action induced by the flow of $X_J$. Finally, let $f : \mathcal M / \mathbb S^1 \to \mathbb R^2$ denote the reduced energy-momentum map, satisfying $\EM = f \circ \rho$. We assume that the fiber $\EM^{-1}(\bar v)$ is a singly pinched torus, which has a saturated neighborhood that contains no other critical points. Then $\EM^{-1}(\bar v) \setminus \{\bar p\}$ is homeomorphic to $\setS^1 \times (0,1)$.

\begin{remark}
  Note that in order to simplify the exposition we consider here only the case where $\EM^{-1}(\bar v)$ contains exactly one focus-focus point. Nevertheless, our approach easily generalizes to the case where $\EM^{-1}(\bar v)$ contains more than one focus-focus point.
\end{remark}

The fibration induced by $f$ onto a neighborhood of $\bar v$ can be decomposed in two parts: a local part defined in a neighborhood $\widetilde B$ of $\rho(\bar p)$, and a part defined in its complement $\widetilde {\mathcal M} \setminus \widetilde B$ where $\widetilde {\mathcal M} = \mathcal M / \mathbb S^1$. We call the first part ``local'' and, with some abuse in terminology, we call the last part ``global''. In these regions the foliation induced by $f$ has the simple structure described in the following proposition, see Figure~\ref{fig/reduced-fibration}.

\begin{figure}
  \includegraphics[width=0.6\columnwidth]{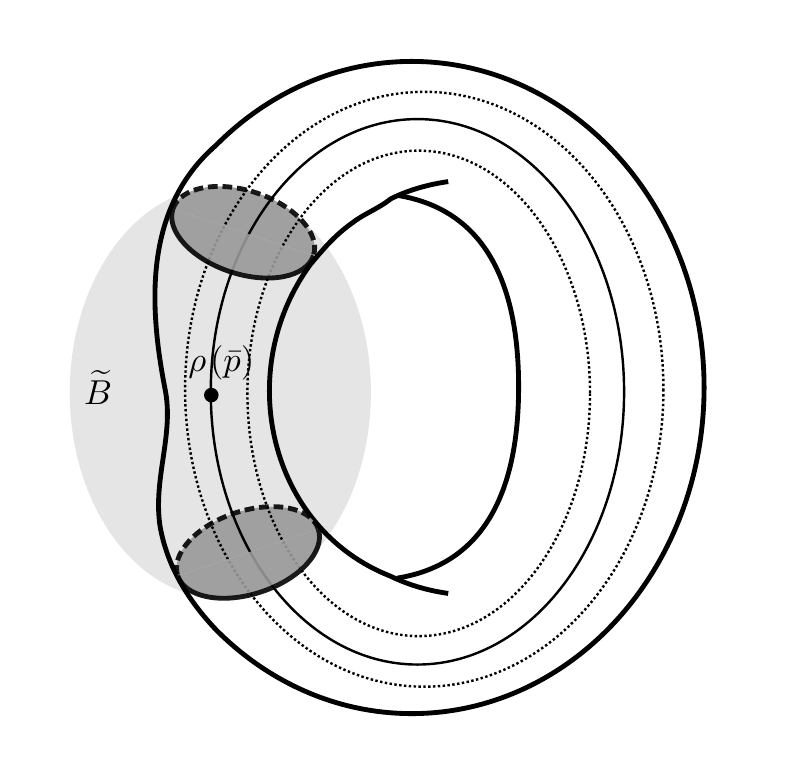}
  \caption{The reduced fibration given by $f$, see Proposition~\ref{local fibration reduced space}. The dotted curves represent regular reduced fibres while the solid curve going through the point $\rho(\bar p)$ represents the reduced pinched torus $f^{-1}(\bar v)$.}
  \label{fig/reduced-fibration}
\end{figure}

\begin{proposition}\label{local fibration reduced space}
  Assume that $\bar p$ is a possibly degenerate focus-focus point and $\bar v = \EM(\bar p)$. Then there exists an open neighborhood $\widetilde B$ of $\rho(\bar p)$ and an open neighborhood $D$ of $\bar v$ such that:
  \begin{enumerate}[label={(\alph*)}]
  \item the fibration
    $f^{-1} (D) \cap \partial \widetilde B \overset{f}{\longrightarrow} D$
    is isomorphic to
    $D \times ( \{ pt \} \sqcup \{ pt \} ) \overset{\mathrm{pr}_1}{\longrightarrow} D$,
    where $\mathrm{pr}_1$ denotes the projection to the first component;

  \item the fibration
    $f^{-1}(D) \setminus \widetilde B \overset{f}{\longrightarrow} D$
    is isomorphic to
    $D \times [0,1] \overset{\mathrm{pr}_1}{\longrightarrow} D$;

  \item the fibration
    $f^{-1}(D^*) \cap \closure{\widetilde B} \overset{f}{\longrightarrow} D^*$ is isomorphic to
    $D^* \times [0,1] \overset{\mathrm{pr}_1}{\longrightarrow} D^*$, where $D^* = D \setminus \{\bar v\}$.
  \end{enumerate}
\end{proposition}

\begin{proof}

  \noindent (a) Since $\EM^{-1}(\bar v) \setminus \{ \bar p \}$ is homeomorphic to a cylinder, its $\mathbb S^1$ reduction, $f^{-1}(\bar v) \setminus \{ \rho(\bar p) \}$ gives an interval $(0,1)$ whose endpoints meet at $\rho(\bar p)$. Moreover, because all objects involved are smooth we have that $f^{-1}(\bar v) \setminus \{ \rho(\bar p) \}$ is a smooth curve in the reduced space. This implies that for any sufficiently small ball $\widetilde B$ around $\rho(\bar p)$ the fiber $f^{-1}(\bar v)$ intersects $\partial \widetilde B$ transversally at exactly two points. Since $f$ is smooth we conclude that for a sufficiently small disk $D$ containing $\bar v$, all fibers $f^{-1}(v)$ for $v \in D$ also intersect $\partial \widetilde B$ transversally at two points.

  \noindent (b) Furthermore, for any $v \in D^* = D \setminus \{ \bar v \}$ we have that $f^{-1}(v)$ is diffeomorphic to $\mathbb S^1$, while as we saw earlier $f^{-1}(\bar v)$ is homeomorphic to $\mathbb S^1$ but smooth outside $\rho(\bar p)$. This implies that for any $v \in D$, $f^{-1}(v) \setminus \widetilde B$ is diffeomorphic to the interval $[0,1]$ and since $D$ is contractible $f^{-1}(D) \setminus \widetilde B$ is isomorphic to $D \times [0,1]$.

  \noindent (c) Finally, $f^{-1}(D^*) \cap \closure{\widetilde B}$ is an orientable fibration with contractible fiber $[0,1]$ over the punctured disk $D^*$. Therefore, it is isomorphic to $D^* \times [0,1]$. The orientability can be explicitly demonstrated by considering the basis $\nabla h$, $\nabla j$, $X_h$ where $f = (h,j)$ and $X_h(\rho(p)) = D\rho(p) X_H(p)$.
\end{proof}

The fibration described in Proposition~\ref{local fibration reduced space} is the projection on $\widetilde{\mathcal M}$ of a fibration of higher dimension defined by the energy-momentum map $\EM$ in the phase space $\mathcal M$. Also in this case the fibers projecting onto a neighborhood of $\bar v$ can be decomposed in a local part and in a global part, and their geometry remains simple.

\begin{proposition}\label{local fibration full space}
  There exists an $\mathbb S^1$-invariant open neighborhood $B$ of the focus-focus point $\bar p$ and an open neighborhood $D$ of $\bar v$ such that:
  \begin{enumerate}[label={(\alph*)}]
  \item the fibration
    $\EM^{-1} (D) \cap \partial B \overset{\EM}{\longrightarrow} D$ is isomorphic to
    $D \times (\mathbb S^1 \sqcup \mathbb S^1) \overset{\mathrm{pr}_1}{\longrightarrow} D$;

  \item the fibration
    $\EM^{-1}(D) \setminus B \overset{\EM}{\longrightarrow} D$ is isomorphic to
    $D \times \mathrm{Cyl} \overset{\mathrm{pr}_1}{\longrightarrow} D$
    where $\mathrm{Cyl}$ is the cylinder $\mathbb S^1 \times [0,1]$;

  \item the fibration $\EM^{-1}(D^*) \cap \closure{B} \overset{\EM}{\longrightarrow} D^*$ is isomorphic to
    $D^* \times \mathrm{Cyl} \overset{\mathrm{pr}_1}{\longrightarrow} D^*$, where $D^* = D \setminus \{\bar v\}$.
  \end{enumerate}
\end{proposition}

\begin{proposition}\label{local transversality}
  If the rotation $1$-form $\vartheta$ is transversal to $\EM$ then $B$ and $D$ can be chosen so that $(F^{-1}(D) \cap \partial B) \cap \Pi = \emptyset$.
\end{proposition}

\begin{proof}
  Recall that $\bar p$, being a fixed point of the $\mathbb S^1$ action, is a pole of $\vartheta$. The transversality condition ensures that we can find a sufficiently small ball $B$ such that $\EM^{-1}(\bar v) \cap B$ contains no other poles of $\vartheta$. Then $\EM^{-1}(D) \cap \partial B$ also contains no poles for a sufficiently small disk $D \ni \bar v$.
\end{proof}

\subsection{Local Variation and Monodromy}
\label{sec/lvam}

Consider the fiber $\EM^{-1}(v)$ for $v$ sufficiently close to $\bar v$ and recall from Proposition~\ref{local fibration full space} that $\EM^{-1}(v) \cap \partial B$ is the disjoint union of two $\mathbb S^1$ orbits $S_- \sqcup S_+$. We make the convention that the flow of $X_H$ in $\EM^{-1}(v) \cap B \simeq \mathrm{Cyl}$ sends points on $S_-$ to $S_+$. For any point $p \in \EM^{-1}(v) \cap B$ let $\gamma_H^\rel(p)$ be the part of the orbit of $X_H$ in $B$ that goes through $p$. Such curve joins a point $p_- \in S_-$ to a point $p_+ \in S_+$. Define
\begin{align}\label{def-Phi-rel}
  \Phi_\rel(v) = \int_{\gamma_H^\rel(p)} \vartheta,
\end{align}
where $\vartheta$ is a rotation 1-form defined in $B \setminus \Pi$ but not necessarily defined globally. In typical situations, $\vartheta$ is the restriction of a global rotation 1-form $\vartheta$ to a neighborhood of the focus-focus point $\bar p$. We further assume that $\vartheta$ is transversal to $\EM$. Note that $\Phi_\rel(v)$ does not depend on the choice of $p \in \EM^{-1}(v)$.

\begin{proposition}\label{local}
  The variation $\var_\Gamma \Phi_\rel$ is equal to $-2k\pi$, where $k \in \mathbb Z$ is the monodromy number of the torus bundle $\EM^{-1}(\Gamma)$.
\end{proposition}

\begin{proof}
  Let $\EM^{-1} (D) \cap \partial B = A_- \cup A_+$ where both $A_+$ and $A_-$ are isomorphic to the solid torus $D \times \mathbb S^1$. Choose a section $\sigma_+: D \to A_+$ and for each point $v \in D$ consider the orbit of $X_H$ through $p = \sigma_+(v)$. Such orbit intersects $A_-$ for the first time at a point $q$. We denote the orbit segment from $p$ to $q$ along the flow of $X_H$ by $\gamma_H^\glob$. This construction defines a section $\sigma_-: D \to A_-$ as the map that sends $v$ to $q$, see Figure~\ref{fig/prop-local}. The triviality of the cylinder bundle $\EM^{-1}(D) \setminus B \to D$ ensures that $\sigma_-$ is well defined. We denote by $[\mathbb S^1]$ the generator of $H_1(A_\pm,\mathbb Z)$ which can be represented by an orbit of the $\mathbb S^1$ action; we have $\int_{[\mathbb S^1]} \vartheta = 2\pi$. Consider now a simple closed path $\Gamma$ in $D$ surrounding the origin. In general, $\sigma_\pm(\Gamma)$ is homologous in $A_\pm$ to $\ell_\pm [\mathbb S^1]$ with $\ell_\pm \in \mathbb Z$. Because $\sigma_+(\Gamma)$ bounds the disk $\sigma_+(D)$ we have $\ell_+ = 0$ and because of the triviality of the cylinder bundle over $D$ we further have $\ell_- = 0$.

\begin{figure}
  \includegraphics[width=0.6\columnwidth]{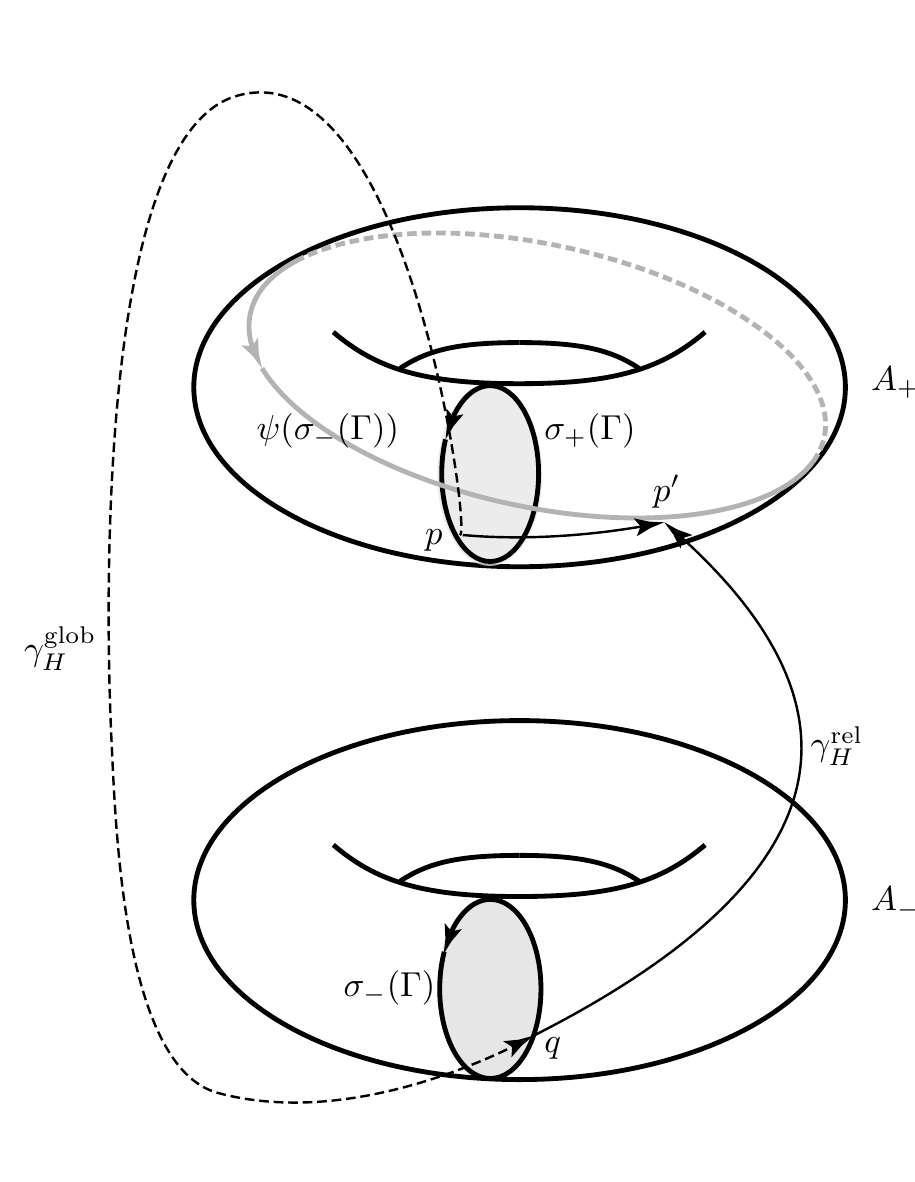}
  \caption{Schematic representation of the solid tori $A_+$ and $A_-$ in the proof of Proposition \ref{local}.}
  \label{fig/prop-local/a}
  \label{fig/prop-local}
\end{figure}

For each point $v \in \Gamma$ follow the flow of $X_H$ from $\sigma_-(v)$ until it reaches $A_+$ for the first time. This defines a map $\psi: \sigma_-(\Gamma) \to A_+$ sending $q$ to $p' = \psi(\sigma_-(v))$. Recall that $\gamma_H^\rel$ is the orbit segment from $q$ to $p'$ along the flow of $X_H$. Consider the cylinder $C$ made up of orbits $\gamma_H^\rel$ as $v$ moves along $\Gamma$. Stokes' theorem gives that $\int_{\partial C} \vartheta$ equals the sum of integrals around cycles $\delta_{ij}$ surrounding the poles of $\vartheta$ on $C$. The latter sum gives $-\var_\Gamma \Phi_\rel$, while $\partial C = - \psi(\sigma_-(\Gamma)) + \sigma_-(\Gamma)$. Therefore
\begin{align*}
  \var_\Gamma \Phi_\rel = \int_{\psi(\sigma_-(\Gamma))} \vartheta - \int_{\sigma_-(\Gamma)} \vartheta.
\end{align*}
Moreover, $\psi(\sigma_-(\Gamma))$ is homologous in $A_+$ to $\ell [\mathbb S^1]$ for some $\ell \in \mathbb Z$. Therefore, since $A_+$ and $A_-$ contain no poles (Proposition~\ref{local transversality}), we have
\begin{align*}
  \var_\Gamma \Phi_\rel = 2 \pi ( \ell - \ell_-) = 2 \ell \pi.
\end{align*}
From the definition of the rotation number it follows that
\begin{align*}
  \psi(\sigma_-(v)) = \varphi_J^{\Theta(v)} (\sigma_+(v)),
\end{align*}
where $\varphi_J^t$ is the time-$t$ flow of $X_J$.
The last relation implies that $\var_\Gamma \Theta = 2\pi(\ell - \ell_+) = 2 \ell \pi$. Therefore, $k = -\ell$ is the monodromy number and we obtain $\var_\Gamma \Phi_\rel = \var_\Gamma \Theta = 2 \ell \pi = - 2 k \pi$ thus concluding the proof.
\end{proof}

\begin{remark}
  If the rotation $1$-form can be extended over a saturated neighborhood of the singular fiber then we can define, in analogy with $\Phi_\rel$, the non-local part of $\Phi$ given by
  \begin{align*}
    \Phi_\glob(v) = \int_{\gamma_H^\glob(p)} \vartheta.
  \end{align*}
  Clearly, $\Phi = \Phi_\rel + \Phi_\glob$. Proposition~\ref{local} implies that $\var_\Gamma \Phi_\glob = 0$. Therefore any contributions to $\var_\Gamma \Phi$ from poles of $\vartheta$ away from $\bar p$ must cancel out. Recall that this situation occurs in the spherical pendulum, see Section~\ref{sec/spherical-pendulum}.
\end{remark}

\subsection{Computation of the Local Variation}\label{sec/1m1}

In the previous section we established that $\var_\Gamma \Phi_\rel = - 2 k \pi$ where $k$ is the monodromy number for the $\mathbb T^2$ bundle over $\Gamma$. In this section we make a specific choice of rotation 1-form near a focus-focus singular point and compute that $\var_\Gamma \Phi_\rel = 2\pi$, thus obtaining $k=-1$.

It is known \cite{Eliasson1984, Miranda2005} that up to using $J$ and possibly replacing $H$ with a function of $H,J$ (operation that does not change the fibration given by $F$), one can assume that in a neighborhood $V$ of a focus-focus singular point $\bar p$ the functions $H,J$ are, in appropriately chosen local symplectic coordinates, the functions
\begin{align}\label{eq/local-nf}
  H = q_1 q_2 - p_1 p_2 \quad \text{and} \quad
  J = \frac{1}{2}(q_1^2 + p_1^2) - \frac{1}{2}(q_2^2 + p_2^2).
\end{align}
with the standard symplectic structure $d p_1 \wedge dq_1 + d p_2 \wedge dq_2$. It can be shown that this is equivalent to the well-studied $A_1$ singularity \cite{Arnold2012}. 

We restrict our attention to an open ball
\begin{align*}
  B = \{ (q_1,p_1,q_2,p_2) \,:\, q_1^2 + q_2^2 + p_1^2 + p_2^2 < 2r \} \subseteq V,
\end{align*}
where $r > 0$ is fixed.
The ball $B$ is $X_J$-invariant but not $X_H$-invariant: for any point in $B$, except for points on the $2$-dimensional stable manifold, the corresponding $X_H$-orbit leaves $B$ in finite time.

In the ball $B$ we make the specific choice of rotation 1-form
\begin{align*}
  \vartheta := d\theta_1 = \frac{p_1 \dd q_1 - q_1 \dd p_1}{q_1^2 + p_1^2}.
\end{align*}

Fibers $\EM^{-1}(j,h) \cap \closure{B}$ for $j^2+h^2 \ne 0$ are diffeomorphic to cylinders $\setS^1 \times [-1,1]$ and the cylinder bundle over $\setR^2 \setminus \{0\}$ is trivial. This is Proposition~\ref{local fibration full space} but it can also be explicitly shown through the following parameterization, which is a symplectic modification of the one given in Ref.~\citenum{Bates2007a}. Specifically, we define the section $\sigma : \setR^2 \to \setR^4$ given by
\begin{align*}
  q_1 = \frac{j + 1}{\sqrt{2}},
  \quad
  p_1 = \frac{h}{\sqrt{2}},
  \quad
  q_2 = \frac{h}{\sqrt{2}},
  \quad
  p_2 = \frac{j - 1}{\sqrt{2}}.
\end{align*}
Note that the section $\sigma$ given here agrees with the section in Ref.~\citenum{Bates2007a} when $h^2+j^2=1$. Then a computation shows that $\sigma^* \omega = 0$, where $\omega$ is the canonical symplectic form in $\setR^4$. Therefore, $\sigma$ is Lagrangian. Furthermore,
\begin{align*}
  H(\sigma(j,h)) = h, \quad J(\sigma(j,h)) = j.
\end{align*}
Then the trivialization is given by
\begin{align*}
  \varphi(u,j,v,h) = \varphi_J^u \circ \varphi_H^v (\sigma(j,h)),
\end{align*}
where $u \in [0,2\pi)$, $v \in \setR$. Here $\varphi_J^t$ and $\varphi_H^t$ represent the time-$t$
flows of $X_J$ and $X_H$ respectively. 

In coordinates $(u,j,v,h)$ the symplectic form becomes
\begin{align*}
  \varphi^* \omega = \dd u \wedge \dd j + \dd v \wedge \dd h,
\end{align*}
while $\varphi^* H = h$, $\varphi^* J = j$.

We further define non-symplectic coordinates $(u,j,w,h)$ by $w = e^{2v} + j$. In the latter coordinates we have
\begin{align*}
  \dd\theta_1 = \dd u + \frac{h\,\dd w - w\,\dd h}{h^2 + w^2}.
\end{align*}

Consider now, in the image of $\EM$, the closed path $\Gamma$ given by
\begin{align*}
  \Gamma(s) = (j(s),h(s)) = (\ell \cos s, \ell sin s), \quad s \in [0,2\pi], \quad \ell > 0.
\end{align*}
The set of $X_H$ orbits
\begin{align*}
  C = \{ (u, j, w, h): u = 0,\, j = \ell \cos s,\, h = \ell \sin s \}
\end{align*}
is a cylinder in phase space containing one $X_H$ orbit (given by $u=0$) for each fiber of $\EM$. The form $\dd\theta_1$ has a pole at one point on $C$, given by $s = \pi$ (therefore, $h=0$, $j=-\ell$) and $w=0$.

\begin{figure}
  \centering
  \hfil\includegraphics[width=0.35\columnwidth]{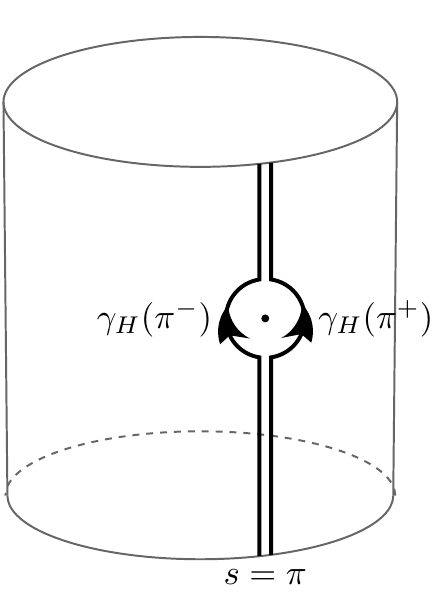}\hfil
  \includegraphics[width=0.35\columnwidth]{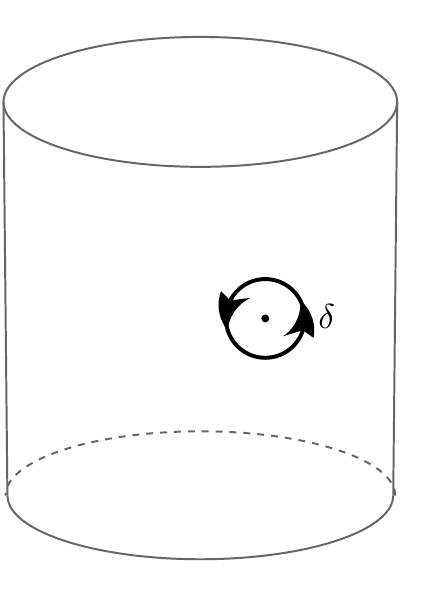}\hfil
  \caption{The variation of the rotation number can be expressed as the integral
    of $\dd\theta_1$ around a single point on the cylinder $C$.}
  \label{fig:vari-rota-numb}
\end{figure}

The relative rotation number $\Phi_\rel$ is discontinuous at the pole and the variation $\var_\Gamma \Phi_\rel$ of the rotation number along $\Gamma$ is the opposite of the size of the discontinuity of $\Phi_\rel$ at $s = \pi$. If we denote the integral curve of $X_H$ on $C$ by $\gamma_H(s)$ then
\begin{align*}
  \var_\Gamma \Phi_\rel = \lim_{s \to \pi^-} \int_{\gamma_H(s)} \dd\theta_1 - \lim_{s \to \pi^+} \int_{\gamma_H(s)} \dd\theta_1.
\end{align*}
Using the fact that $\dd\theta_1$ is closed we can then express $\var_\Gamma \Phi_\rel$ as
\begin{align*}
  \var_\Gamma \Phi_\rel = - \int_\delta \dd\theta_1,
\end{align*}
where $\delta$ is a positively oriented closed path on $C$ winding once around the pole. Using coordinates $(w,h)$ in a neighborhood of the pole $C$ we find that
\begin{align*}
  \dd\theta_1|_C = \frac{h \,\dd w - w\,\dd h}{h^2 + w^2}.
\end{align*}
The choice of the coordinates as $(w,h)$ is so that the orientation is the same
as the choice $(s,w)$ used in defining the positive direction for
$\delta$. Then
\begin{align*}
  \int_\delta \frac{h \,\dd w - w\,\dd h}{h^2 + w^2} = - 2\pi,
\end{align*}
which gives
\begin{align*}
  \var_\Gamma \Phi_\rel = 2 \pi,
\end{align*}
and therefore the monodromy number is $k=-1$.

\begin{remark}
  In the local trivialization we could have chosen $\vartheta = \dd u$ as a rotation $1$-form. Such a choice would give $\Phi_\rel = 0$ and thus $\var_\Gamma \Phi_\rel = 0$ thus contradicting our result. Nevertheless, one can check that $\dd u $ is not transversal to $\EM$ and therefore Theorem \ref{thm/main} is not applicable with this choice of rotation $1$-form.
\end{remark}

\section{Noncompact Fibrations and Scattering Monodromy}
\label{sec/scattering}
In the previous sections we have been considering integrable Hamiltonian systems with an $\mathbb S^1$ action and connected, compact fibers. In this section we discuss the case of systems with noncompact fibers and we show how the local considerations in Section~\ref{sec/local} lead to a definition of monodromy for such systems. This noncompact monodromy is then compared to the notion of scattering monodromy.

\subsection{Definition of Noncompact Monodromy}

We assume that the integral map $\EM$ is $\mathbb S^1$ invariant and that the fibers are connected but noncompact while the flow of $X_H$ is complete. Under these assumptions, the regular fibers of $\EM$ are cylinders $\mathbb S^1 \times \mathbb R$.

Such systems can be claimed to ``have no monodromy'' in the following sense, see Ref.~\citenum{Bates2007a}. Consider a simple closed path $\Gamma$ in the set of regular values of $\EM$. Then $\EM^{-1}(\Gamma) \overset{\EM}{\longrightarrow} \Gamma$ is an orientable cylinder bundle over $\Gamma$ and it is thus trivial, that is, isomorphic to the bundle $\mathbb S^1 \times (\mathbb S^1 \times \mathbb R) \overset{\mathrm{pr_1}}{\longrightarrow} \mathbb S^1$.

Expanding on the concept of scattering monodromy, introduced in Ref.~\citenum{Bates2007a}, we propose to define noncompact monodromy by appropriately identifying the two ``ends'' of the cylinder fibers of $\EM$, turning the cylinder bundle over $\Gamma$ into a torus bundle. We show that our definition of noncompact monodromy, which is based on topological considerations, matches scattering monodromy for the particular system studied in Ref.~\citenum{Bates2007a} and we explain in detail how the two concepts are related.

First, we construct a torus bundle starting from $\EM^{-1}(\Gamma) \overset{\EM}{\longrightarrow} \Gamma$, through the following procedure. Recall that we consider an $\mathbb S^1$ invariant integral map $\EM$, with noncompact, connected, fibers. Let $\Gamma$ be a simple closed path in the set $\mathcal R$ of regular values of $\EM$, bounding a disk $D \subseteq \overline {\mathcal R}$. Denote by $f$ the $\mathbb S^1$-reduced integral map, that is $\EM = f \circ \rho$, where $\rho$ is the reduction map of the $\mathbb S^1$ action.

We make the assumption that for each $v \in D$ the fiber $f^{-1}(v) = \EM^{-1}(v) / \mathbb S^1$ is homeomorphic to $\mathbb R$. Then there is a homeomorphism $g: \EM^{-1}(D) / \mathbb S^1 \to D \times \mathbb R$ such that $\mathrm{pr}_1 \circ g = f$, that is, the bundles $\EM^{-1}(D) / \mathbb S^1 \overset{f}{\longrightarrow} D$ and $D \times \mathbb R \overset{\mathrm{pr}_1}{\longrightarrow} D$ are topologically isomorphic. Let $B_m = g^{-1}(D \times [-m,m]) \subset f^{-1}(D)$, $m > 0$. We have $\partial B_m = N_{+m} \cup N_{-m}$ with $N_{\pm m}$ homeomorphic to $D$. Let $A_{\pm m} = \rho^{-1}(N_{\pm m})$. Assuming that the $\mathbb S^1$ action has a finite number of fixed points, we can choose $m$ large enough so that all such fixed points in $\EM^{-1}(D)$ are contained in $\rho^{-1}(B_m)$. Thus, each set $A_{\pm m}$ is homeomorphic to a solid torus. Consider continuous sections
\begin{align*}
  \sigma_{\pm m} : D \simeq N_{\pm m} \to A_{\pm m}
\end{align*}
of the principal $\mathbb S^1$ bundle $A_{\pm m} \overset{f}{\longrightarrow} D$. Then define homeomorphisms
\begin{align*}
  h_{\pm m} : D \times \mathbb S^1 \to A_{\pm m} : (v,t) \mapsto \varphi_J^t(\sigma_{\pm m}(v)),
\end{align*}
which in turn allow the definition of the identification map
\begin{align*}
  \eta_m : A_{+m} \to A_{-m} : p \mapsto \eta_m (p) = h_{-m} \circ (h_{+m})^{-1}(p).
\end{align*}
Consider now the space
\begin{align*}
  C_m = \rho^{-1}(B_m) / \sim_{\eta_m},
\end{align*}
obtained by identifying points at the boundary $A_{+m} \cup A_{-m}$ of $\rho^{-1}(B_m)$ through the map $\eta_m$. This construction turns $\rho^{-1}(B_m) \subset \EM^{-1}(D)$ into a closed topological manifold without boundary. Finally, define
\begin{align*}
  T_m = (\rho^{-1}(B_m) \cap \EM^{-1}(\Gamma)) / \sim_{\eta_m} = \rho^{-1}(g^{-1}(\Gamma \times [-m,m])) / \sim_{\eta_m}.
\end{align*}
The space $T_m$ is a bundle of tori over $\Gamma$ with the projection map given by $\EM$. Note that $\EM$ is well-defined on $T_m$ since $\EM \circ \eta_m = \EM$. We will denote the bundle $T_m \overset{\EM}{\longrightarrow} \Gamma$ simply by $T_m$. With this construction we can now give the following definition of noncompact monodromy.

\begin{definition}\label{noncompact}
  The \emph{noncompact monodromy} of the cylinder bundle $\EM^{-1}(\Gamma) \overset{\EM}{\longrightarrow} \Gamma$ is the monodromy of the torus bundle $T_m$ if there is $M > 0$ such that the monodromy of $T_m$ is constant for all $m > M$.
\end{definition}

\begin{remark}
  In the construction of the torus bundle above we assumed that the identification of the ends of the ``cut'' cylinders is done over the whole disk $D$ bounded by $\Gamma$. This is essential for defining the torus bundle $T_m$ uniquely (up to isotopy). If the identification is given only over $\Gamma$ then there is enough freedom to construct torus bundles with arbitrary monodromy number. Note that the identification of $A_+$ and $A_-$ over the whole $D$, and not only over $\Gamma$, also plays an essential role in the proof of Proposition~\ref{local}.
\end{remark}

In the construction of the torus bundle $T_m$ one can define, in analogy with our local description in Section~\ref{sec/local}, and the definition of $\Phi^\rel$,
\begin{align*}
  \Phi_m(v) = \int_{\gamma_H(v)} \vartheta,
\end{align*}
where $\vartheta$ is any rotation $1$-form and $\gamma_H(v)$ is an orbit segment of $X_H$ going from $A_{-m}$ to $A_{+m}$. Then we define 
\begin{align}\label{eq/def-rotation-number-scattering}
  \Phi(v) = \lim_{m \to \infty} \Phi_m(v),
\end{align}
if the latter limit exists.

Moreover, one can define the rotation number $\Theta(v)$ in the following way. Let $q = \sigma_{-m}(v)$ and consider the $X_H$ orbit that starts at $q$ end ends at a point $p'=\psi_m(q) = \psi_m(\sigma_{-m}(v))$ in $A_{+m}$. Furthermore, let $p = \eta_m^{-1}(q) = \sigma_{+m}(v)$. Then define $\Theta_m(v)$ by
\begin{align*}
  \psi_m(\sigma_{-m}(v))=\varphi_J^{\Theta_m(v)}(\sigma_{+m}(v)),
\end{align*}
and, finally, the rotation number is defined by
\begin{align}\label{eq/def-geom-rotation-number-scattering}
  \Theta(v) = \lim_{m \to \infty} \Theta_m(v),
\end{align}
if the latter limit exists.

\subsection{Comparison to Scattering Monodromy}

We now focus on the integrable Hamiltonian system given by $\EM$ in Eq.~\eqref{eq/local-nf}, but with $\EM$ now defined over $\mathbb R^4$, and not only in a neighborhood of the origin as in Section~\ref{sec/1m1}. Note that in this case the regular fibers of $\EM$ are cylinders. Moreover, this is precisely the system studied in Ref.~\citenum{Bates2007a} up to assigning coordinates $(q_1,p_1,q_2,p_2)$ to $(x_1,x_2,y_1,-y_2)$ respectively.

\begin{proposition}
  The noncompact monodromy of $\EM$ in Eq.~\eqref{eq/local-nf} is the local monodromy determined in Section~\ref{sec/1m1} and it thus has monodromy number $k=-1$.
\end{proposition}

\begin{proof}
  Comparing the proof of monodromy in Proposition~\ref{local} for the bundle $\EM^{-1}(\Gamma) \overset{\EM}{\longrightarrow}\Gamma$ with the construction of the torus bundle $T_m$, one sees that for finite $m > 0$ the two bundles have the same monodromy. Since the monodromy of $T_m$ is the same for all $m > 0$ the noncompact monodromy is well-defined and equal to the local monodromy.
\end{proof}

Further note that the noncompact monodromy defined here coincides with the scattering monodromy introduced in Ref.~\citenum{Bates2007a}. This is not a coincidence; instead, it sheds light to an aspect of the definition of scattering monodromy, that is, an implicit compactification of the cylinder bundle in Ref.~\citenum{Bates2007a}.

Let us recall the definition of scattering monodromy from Ref.~\citenum{Bates2007a}. The projections of the integral curves of $X_H$ in this system have well defined asymptotic directions in the $(q_1,p_1)$-plane as $t \to \pm\infty$. The angle between these two directions on the fiber $\EM^{-1}(v)$ is given by the \emph{scattering angle} $\Theta_s(v)$ which is shown to equal $\arg(j+ih)$, where $v=(h,j)$. It follows that $\var_\Gamma \Theta_s = 2\pi$ and therefore one gets a non-trivial variation which is reminiscent of the variation of the rotation number in compact monodromy. The scattering angle can also be obtained by integrating the form
\begin{align*}
  \vartheta_s = \frac{q_1 \, dp_1 - p_1 \, dq_1}{q_1^2 + p_1^2},  
\end{align*}
along an integral curve of $X_H$, or its projection in the $(q_1,p_1)$-plane. Note that, following our terminology, $\vartheta_s$ is a rotation $1$-form in $\mathbb R^4 \setminus \{ q_1=p_1=0 \}$, cf.~Definition~\ref{defrotform}.

Comparing definitions one directly sees that, choosing the rotation $1$-form to be $\vartheta=\vartheta_s$, we have $\Phi(v) = \Theta(v) = \Theta_s(v)$, where $\Phi(v)$ and $\Theta(v)$ are given in Eqs.~\eqref{eq/def-rotation-number-scattering} and \eqref{eq/def-geom-rotation-number-scattering} respectively. Moreover, note that for any $m > 0$ we have, using the result of Section~\ref{sec/1m1}, that
\begin{align*}
  \var_\Gamma \Phi_m = 2\pi.
\end{align*}
Therefore, we also have
\begin{align*}
  \var_\Gamma \Phi  = 2\pi,
\end{align*}
in accordance with $\var_\Gamma \Theta_s = 2\pi$ obtained in Ref.~\citenum{Bates2007a}.

Note that there are two aspects of this story. One computational, where we define an \emph{ad hoc} rotation number, and one topological, where we identify the two ends of the cylinders in the bundle $\EM^{-1}(\Gamma)$. The computational considerations show that scattering monodromy is precisely the local monodromy considered in Section~\ref{sec/local}. The construction of the torus bundle $T_m$ described earlier and culminating to Definition~\ref{noncompact} of noncompact monodromy gives a topological interpretation of this computation.

In particular, for the construction of the torus bundle $T_m$ for $\EM$ in Eq.~\eqref{eq/local-nf} we define the sections $\sigma_{\pm m}: D \to A_{\pm m}$ by
\begin{align*}
  \sigma_{\pm m}(h,j) &= (q_1+ip_1, q_2+ip_2) \\
               &=
  \Biggl( \bigl[ (j^2 + h^2 + m^2)^{1/2} + j \bigr]^{1/2},
  \biggl[ \frac{(j^2 + h^2 + m^2)^{1/2} - j}{h^2 + m^2} \biggr]^{1/2} (h \mp im) 
  \Biggr).
\end{align*}
Note that for large values of $m$, $m^2 \gg h^2 + j^2$, the sections are asymptotically equal to
\begin{align*}
  \sigma_{\pm m}(h,j) = (q_1, p_1, q_2, p_2) \simeq \left( \sqrt{m}, 0, 0, \mp \sqrt{m} \right).
\end{align*}
This choice corresponds exactly to asymptotic motion along the $q_1$ axis while the sign of $p_2 = -X_H(q_1) = \pm \sqrt{m}$ signifies incoming (for `$+$') or outgoing (for `$-$') motions. Using the terminology of Ref.~\citenum{Bates2007a}, incoming asymptotic motions are ``negative ends'' and outgoing asymptotic motions are ``positive ends'' of integral curves of $X_H$. From this point of view, the map $\eta_m$ for large $m$ identifies negative ends with positive ends, that is, it identifies incoming and outgoing motions that are asymptotically along the same direction in the $(q_1,p_1)$-plane.

Recall from Section~\ref{sec/lvam} that $\Phi_m^\rel(v)$ measures the time to go along the flow of $X_J$ from $\sigma_{+m}(v)$ to $\psi(\sigma_{-m}(v))$. In scattering monodromy, $\Theta_s(v)$ measures the time to go along the flow of $X_J$ from the ``parallel transport of the negative end of the integral curve'' to the ``positive end of the integral curve''. The parallel transport defined in Ref.~\citenum{Bates2007a} provides an identification of ``positive ends'' and ``negative ends'' that extends over all fibers, including the singular one. This demonstrates that the computationally defined scattering monodromy is indeed the monodromy of a properly defined torus bundle.

A different, but closely related to Ref.~\citenum{Bates2007a}, definition of scattering monodromy is given in Ref.~\citenum{Dullin2008}. A \emph{deflection angle} $\Delta\phi$ is defined as the difference between the asymptotic directions of two outgoing classical trajectories: one for a free particle (that is, without scattering) and one scattered under a repulsive potential corresponding to Hamiltonian function
\begin{align*}
  H(q,p)=\frac12(p_1^2+p_2^2)-\frac12(q_1^2+q_2^2),
\end{align*}
invariant under the circle action generated by the angular momentum $J=q_1p_2-q_2p_1$. We note here that a linear symplectic coordinate change brings this system to the form given in Eq.~\eqref{eq/local-nf}. It is then shown in Ref.~\citenum{Dullin2008} that for a closed path $\Gamma$ around the origin in the $(H,J)$-plane the map $\chi: \Gamma \simeq \mathbb S^1 \to \mathbb S^1$ which assigns to each point on $\Gamma$ the corresponding deflection angle $\Delta\phi$ has degree $1$, and this observation is characterized as scattering monodromy. The correspondence between the result in Ref.~\citenum{Dullin2008} and the noncompact monodromy defined here is straightforward. The identification of the two ``ends'' of the fibers of $F$ (typically, cylinders) is given by the flow of the reference Hamiltonian $H_0=\tfrac12(p_1^2+p_2^2)$ corresponding to free motion. Such identification extends continuously to the singular fiber $F^{-1}(0)$ since the flow of $H_0$ has no fixed points. Given such identification the map $\chi$ coincides with the map $\Theta$ defined in Eq.~\eqref{eq/def-geom-rotation-number-scattering} and as the proof of Proposition~\ref{local} shows, the degree of $\Theta$ determines the monodromy of the corresponding torus bundle.
  
The results obtained in the present paper show that the descriptions of scattering monodromy in Ref.~\citenum{Bates2007a} and in Ref.~\citenum{Dullin2008} are ultimately equivalent. The difference between the two approaches boils down to a different identification of incoming and outgoing asymptotic directions. Since both identifications extend continuously inside the disk bounded by the closed path $\Gamma$ we conclude that they define the same monodromy which coincides with the \emph{noncompact monodromy} as given by Definition~\ref{noncompact}.

The main difference between our approach and the one in Refs.~\citenum{Bates2007a} and \citenum{Dullin2008} is the following. We treat the choice of the rotation 1-form and the choice of identification of the cylinder ends as independent. In particular, the choice of rotation 1-form determines the function $\Phi(v)$ while the choice of identification determines the function $\Theta(v)$; it turns out that $\var_\Gamma \Theta = \var_\Gamma \Phi$. Refs. \citenum{Bates2007a} and \citenum{Dullin2008} implicitly give an identification of cylinder ends and then choose to integrate a rotation 1-form for which $\Phi(v)=\Theta(v)$.

\section{Discussion}
\label{sec/discussion}

Given a 2DOF integrable system with a circle symmetry, we introduced the concept of rotation 1-form. A rotation 1-form is a closed 1-form which measures the displacement with respect to the flow of the circle action (just as a connection 1-form for principal bundles). This displacement is strictly related to the classical \emph{rotation number}, whose multivaluedness gives the integer that is commonly referred to as the \emph{monodromy}. In presence of singularities of the integrable system, the rotation 1-form is necessarily undefined in some submanifold of poles in the phase space. It is precisely a residue-like formula around such poles that gives the monodromy.

Under the hypothesis that the system has a simple focus-focus singularity, the computation turns out to give the monodromy when confined in a neighborhood of the singularity. This allows us to define monodromy also in the case of non-compact fibration and to show that it extends previous notions of scattering monodromy.

Being associated to general closed 1-forms with poles, the sum of the residues is not necessarily an integer multiple of $2\pi$. Our idea appears to be suitable to generalizations to more general types of monodromy, such as fractional monodromy. The presence of threads of singularities makes the analysis more complicated, and requires a deeper investigation, that might involve Picard-Lefschetz formulas.

\section*{Acknowledgments}

This work was partially supported by the Jiangsu University Natural Science Research Program grant 13KJB110026. The authors would like to thank the Xi'an Jiaotong-Liverpool University and the University of Catania, where part of this research took place, for their hospitality. We acknowledge support from the National Natural Science Foundation of China (Grant no. 61502132), ANR-DFG research program Explosys (ANR-14-CE35-0013-01; DFG-Gl 203/9-1), Mexican UNAM project PREI 2016-II, and Franco-Mexican project LAISLA.

We thank the referee for the constructive remarks which helped clarify the presentation of the results.

\bibliographystyle{plain}
\bibliography{local-monodromy}

\end{document}